\documentclass[12pt]{article}
\usepackage{amsmath}
\usepackage{times}
\usepackage{graphicx}
\usepackage{color}
\usepackage{multirow}
\usepackage[authoryear]{natbib}
\usepackage{rotating}
\usepackage{bbm}
\usepackage{latexsym}
\usepackage{amssymb,amsthm}			            
\usepackage{bm,xspace,fancyhdr}					
\usepackage[tight,footnotesize]{subfigure}
\usepackage{epstopdf}                           
\usepackage{pdfsync}				            
\usepackage{setspace}
\graphicspath{{./figs/}}

\newcommand{\xmath}[1]	{\ensuremath{#1}\xspace}

\newcommand{\reals}	{\xmath{\mathbb{R}}}
\newcommand{\norm}[1]	{\xmath{\left|\left| {#1} \right|\right|}}

\newcommand{\expect}[1]	{\xmath{\mathcal{E}\left[{#1}\right]}}

\newcommand{\spnum}	{\xmath{K}} 
\newcommand{\spnumf}	{\xmath{K^{\ast}}} 
\newcommand{\nodenum}	{\xmath{M}}
\newcommand{\memlen}	{\xmath{N}}

\newcommand{\memlenf}	{\xmath{N^{\ast}}}
\newcommand{\conmat}	{\xmath{\bm{W}}}
\newcommand{\conmatev}	{\xmath{\bm{U}}}
\newcommand{\conmated}	{\xmath{\bm{D}}}

\newcommand{\ffvec}	{\xmath{\bm{z}}}

\newcommand{\ffvecu}	{\xmath{\widetilde{\bm{z}}}}
\newcommand{\tvar}	{\xmath{n}}
\newcommand{\nodevar}	{{x}}
\newcommand{\nodevec}	{\xmath{\bm{\nodevar}}}
\newcommand{\nodevect}[1]	{\xmath{\bm{\nodevar}[#1]}}

\newcommand{\sigvar}	{{s}}
\newcommand{\sigmax}	{\xmath{\sigvar_{\max}}}
\newcommand{\sigvec}	{\xmath{\bm{\sigvar}}}
\newcommand{\sigvest}	{\xmath{\widehat{\bm{\sigvar}}}}

\newcommand{\sigvart}[1]	{\xmath{\sigvar[#1]}}

\newcommand{\ripmat}	{\xmath{\bm{A}}}
\newcommand{\fftmat}	{\xmath{\bm{F}}}
\newcommand{\decvar}	{\xmath{q}}
\newcommand{\decmat}	{\xmath{\bm{Q}}}
\newcommand{\wavmat}	{\xmath{\bm{\Psi}}}
\newcommand{\zmat}	{\xmath{\widetilde{\bm{Z}}}}

\newcommand{\ripcd}	{\xmath{\delta}}
\newcommand{\ripcc}	{\xmath{C}}
\newcommand{\ripex}	{\xmath{\gamma}}

\providecommand{\defstart}{
\begin{tabular}{l|l|l} Macro & Symbol & Meaning \\ \hline }
\providecommand{\defstop}{\end{tabular}}

\newtheorem{lemma}{Lemma}[section]
\newtheorem{thm}{Theorem}[section]

\newcommand{\comps}	{\xmath{\mathbb{C}}}
\newcommand{\proba}[1]	{\xmath{\mathcal{P}\left[{#1}\right]}}

\newcommand{\captionfonts}{\normalsize}

\makeatletter  
\long\def\@makecaption#1#2{%
  \vskip\abovecaptionskip
  \sbox\@tempboxa{{\captionfonts #1: #2}}%
  \ifdim \wd\@tempboxa >\hsize
    {\captionfonts #1: #2\par}
  \else
    \hbox to\hsize{\hfil\box\@tempboxa\hfil}%
  \fi
  \vskip\belowcaptionskip}
\makeatother   

\begin{document}
\doublespace

\hspace{13.9cm}1

\ \vspace{20mm}\\

{\LARGE Short Term Memory Capacity in Networks via the Restricted Isometry Property}

\ \\
{\bf \large Adam S. Charles, Han Lun Yap, Christopher J. Rozell}\\
{School of Electrical and Computer Engineering, Georgia Institute of Technology, Atlanta, GA}\\
%
%


\thispagestyle{empty}
\markboth{}{NC instructions}

\begin{center} {\bf Abstract} \end{center}
Cortical networks are hypothesized to rely on transient network activity to support short term memory (STM).  In this paper we study the capacity of randomly connected recurrent linear networks for performing STM when the  input signals are approximately sparse in some basis.  We leverage results from compressed sensing to provide rigorous non-asymptotic recovery guarantees,  quantifying the  impact of the input sparsity level, the input sparsity basis, and the network characteristics on the system capacity.
Our analysis demonstrates that network memory capacities can scale superlinearly with the number of nodes, and in some situations can achieve STM capacities that are much larger than the network size.  We provide perfect recovery guarantees for finite sequences and recovery bounds for infinite sequences.  The latter analysis predicts that network STM systems may have an optimal recovery length that balances errors due to omission and recall mistakes. 
Furthermore, we show that the conditions yielding optimal STM capacity can be embodied in several network topologies, including networks with sparse or dense connectivities.

\section{Introduction}
\label{sec:intro}

-------------------------

Short term memory (STM) is critical for neural systems to understand non-trivial environments and perform complex tasks.  While individual neurons could potentially account for very long or very short stimulus memory (e.g., through changing synaptic weights or membrane dynamics, respectively), useful STM on the order of seconds is conjectured to be due to transient \emph{network} activity.  Specifically, stimulus perturbations can cause activity in a recurrent network long after the input has been removed, and recent research hypothesizes that cortical networks may rely on transient activity to support STM~\citep{JAE:2004,MAA:2002,buonomano2009state}.  

Understanding the role of memory in neural systems requires determining the fundamental limits of STM capacity in a network and characterizing the effects on that capacity of the network size, topology, and input statistics.  Various approaches to quantifying the STM capacity of linear~\citep{JAE:2001,SOM:2004,SOM:2008,hermans2010} and nonlinear~\citep{latham2013} recurrent networks have been used, often assuming Gaussian input statistics~\citep{JAE:2001,SOM:2004,hermans2010,latham2013}. These analyses show that even under optimal conditions, the STM capacity (i.e., the length of the stimulus able to be recovered) scales only linearly with the number of nodes in the network.  While conventional wisdom holds that signal structure could be exploited to achieve more favorable capacities, this idea has generally not been the focus of significant rigorous study.

Recent work in computational neuroscience and signal processing has shown that many signals of interest have statistics that are strongly non-Gaussian, with low-dimensional structure that can be exploited for many tasks.  In particular, sparsity-based signal models (i.e., representing a signal using  relatively few non-zero coefficients in a basis) have recently been shown to be especially powerful.  In the computational neuroscience literature, sparse encodings increase the capacity of associative memory models~\citep{BAU:1988} and are sufficient neural coding models to account for several properties of neurons in primary visual cortex (i.e., response preferences~\citep{OLS:1996} and nonlinear modulations~\citep{ZHU:2012}).  In the signal processing literature, the recent work in compressed sensing (CS)~\citep{CAN:2006b,SOM2012review} has established strong guarantees on sparse signal recovery from highly undersampled measurement systems.  

\citet{GAN:2010} have previously conjectured that the ideas of CS can be used to achieve STM capacities that exceed the number of network nodes in an orthogonal recurrent network when the inputs are sparse in the canonical basis (i.e. the input sequences have temporally localized activity).  While these results are compelling and provide a great deal of intuition, the theoretical support for this approach remains an open question as the results in~\citep{GAN:2010} use an asymptotic analysis on an approximation of the network dynamics to support empirical findings.  In this paper we establish a theoretical basis for CS approaches in network STM by providing rigorous non-asymptotic recovery error bounds for an exact model of the network dynamics and input sequences that are sparse in any general basis (e.g., sinusoids, wavelets, etc.).  Our analysis shows conclusively that the STM capacity can scale superlinearly with the number of network nodes, and quantifies the  impact of the input sparsity level, the input sparsity basis, and the network characteristics on the system capacity.  We provide both perfect recovery guarantees for finite inputs, as well as bounds on the recovery performance when the network has an arbitrarily long input sequence.  The latter analysis predicts that network STM systems based on CS may have an optimal recovery length that balances errors due to omission and recall mistakes.  Furthermore, we show that the structural conditions yielding optimal STM capacity in our analysis can be embodied in many different network topologies, including networks with both sparse and dense connectivities.



\section{Background}
\subsection{Short Term Memory in Recurrent Networks}

Since understanding the STM capacity of networked systems would lead to a better understanding of how such systems perform complex tasks, STM capacity has been studied in several network architectures, including  discrete-time networks~\citep{JAE:2001,SOM:2004,SOM:2008}, continuous-time networks~\citep{hermans2010,busing2010connectivity}, and spiking networks~\citep{MAA:2002,mayor2005signal,legenstein2007edge,latham2013}.  While many different analysis methods have been used, each tries to quantify the amount of information present in the network states about the past inputs.  For example, in one approach taken to study echo state networks (ESNs)~\citep{SOM:2004,SOM:2008,hermans2010}, this information preservation is quantified through the correlation between the past input and the current state. When the correlation is too low, that input is said to no longer be represented in the state. The results of these analyses conclude that for Gaussian input statistics, the number of previous inputs that are significantly correlated with the current network state is bounded by a linear function of the network size.

In another line of analysis, researchers have sought to directly quantify the degree to which different inputs lead to unique network states~\citep{JAE:2001,MAA:2002,legenstein2007edge,strauss2012design}.  In essence, the main idea of this work is that a one-to-one relationship between input sequences and the network states should allow the system to perform an inverse computation to recover the original input.  A number of specific properties have been proposed to describe the uniqueness of the network state with respect to the input.   In spiking liquid state machines (LSMs), in work by~\cite{MAA:2002}, a separability property is suggested that guarantees distinct network states for distinct inputs and follow up work~\citep{legenstein2007edge} relates the separability property to practical computational tasks through the Vapnik-Chervonenkis (VC) dimension~\citep{vapnik1971uniform}.  More recent work analyzing similar networks using separation properties~\citep{latham2013,busing2010connectivity} gives an upper bound for the STM capacity that scales like the logarithm of the number of network nodes. 

In discrete ESNs, the echo-state property (ESP) ensures that every network state at a given time is uniquely defined by some left-infinite sequence of inputs~\citep{JAE:2001}. The necessary condition for the ESP is that the maximum eigenvalue magnitude of the system is less than unity (an eigenvalue with a magnitude of one would correspond to a linear system at the edge of instability).  While the ESP ensures uniqueness, it does not ensure robustness and output computations can be sensitive to small perturbations (i.e., noisy inputs). A slightly more robust property  looks at the conditioning of the matrix describing how the system acts on an input sequence~\citep{strauss2012design}.  The condition number describes not only a one-to-one correspondence, but also quantifies how small perturbations in the input affect the output.   While work by~\cite{strauss2012design} is closest in spirit to the analysis in this paper, it ultimately concludes that the STM capacity still scales linearly with the network size.

Determining whether or not a system abides by one of the separability properties depends heavily on the network's construction. In some cases, different architectures can yield very different results.   For example, in the case of randomly connected spiking networks, low connectivity  (each neuron is connected to a small number of other neurons) can lead to large STM capacities~\citep{legenstein2007edge,busing2010connectivity}, whereas high connectivity leads to chaotic dynamics and smaller STM capacities~\citep{latham2013}.  In contrast, linear ESNs with high connectivities (appropriately normalized)~\citep{busing2010connectivity}  can have relatively large STM capacities (on the order of the number of nodes in the network)~\citep{SOM:2008,strauss2012design}.   Much of this work centers around using systems with orthogonal connectivity matrices, which leads to a topology that robustly preserves information. Interestingly, such systems can be constructed to have arbitrary connectivity while preserving the information preserving properties~\citep{strauss2012design}.

While a variety of networks have been analyzed using the properties described above, these analyses ignore any structure of the inputs sequences that could be used to improve the analysis~\citep{JAE:2001,mayor2005signal}. Conventional wisdom has suggested that STM capacities could be increased by exploiting structure in the inputs, but formal analysis has rarely addressed this case.  For example, work by~\cite{GAN:2010} builds significant intuition for the role of structured inputs in increasing STM capacity, specifically proposing to use the tools of CS to study the case when the input signals are temporally sparse.  However, the analysis by~\cite{GAN:2010} is asymptotic and focuses on an annealed (i.e., approximate) version of the system that neglects correlations between the network states over time.  The present paper can be viewed as a generalization of this work to provide formal guarantees for STM capacity of the exact system dynamics, extensions to arbitrary orthogonal sparsity bases, and recovery bounds when the input exceeds the capacity of the system (i.e., the input is arbitrarily long).

\begin{figure}[t]
	\begin{center}
		\includegraphics[width=5in]{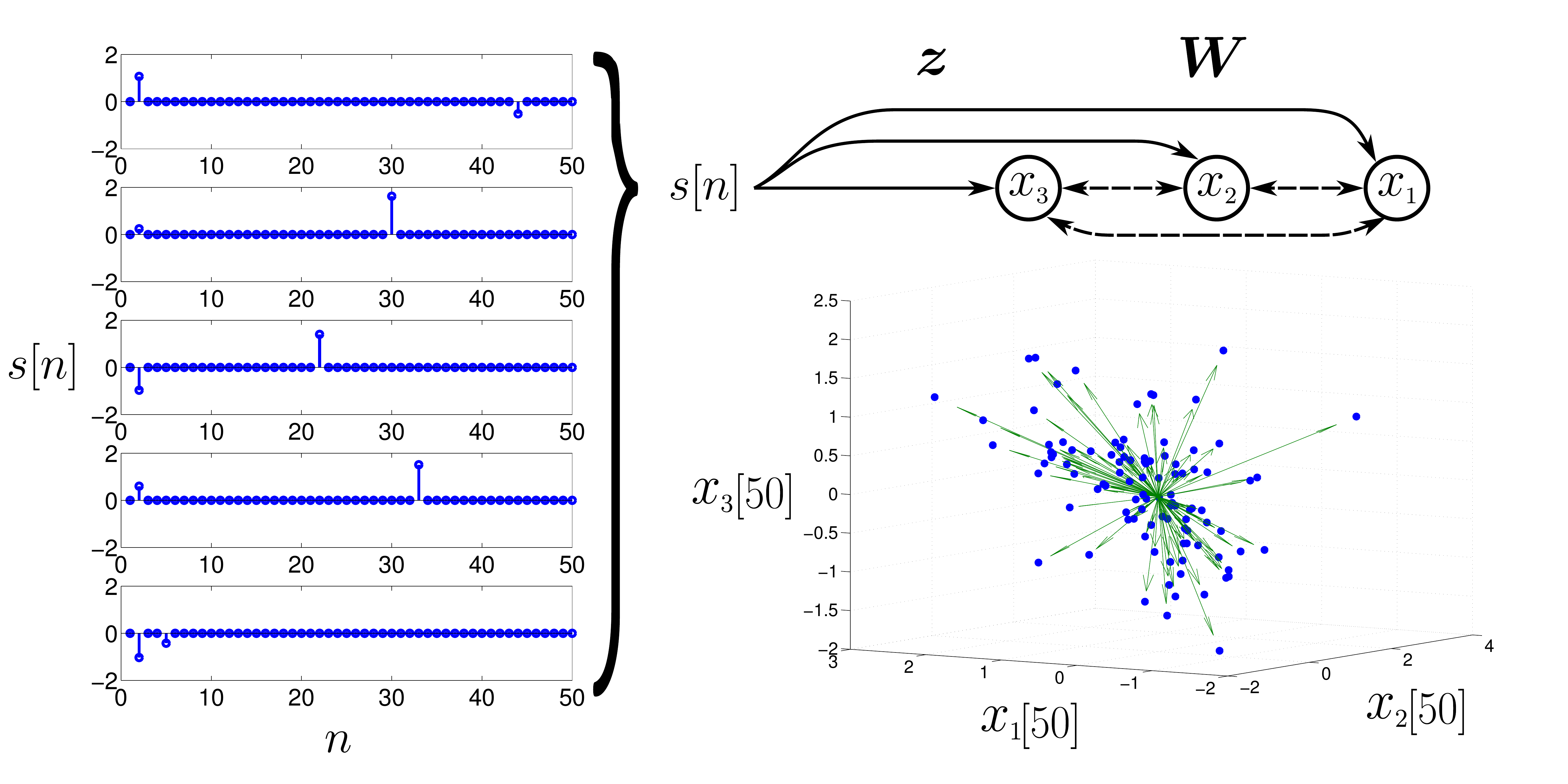}
	\end{center}
	\caption{The current state of the network encodes information about the stimulus history. Different stimuli (examples shown to the left), when perturbing the same system (in this figure a three neuron orthogonal network) result in distinct states \nodevec = $[x_1, x_2, x_3]^T$ at the current time ($n=50$). The current state is therefore informative for distinguishing between the input sequences.}
	\label{fig:NNfig}
\end{figure}

\subsection{Compressed Sensing}
\label{sec:CS}

In the CS paradigm, a signal $\sigvec\in\reals^{\memlen}$ is sparse in a basis $\wavmat$ so that it can be approximately written as $\sigvec \approx \wavmat\bm{a}$, where most of the entries in $\bm{a}\in\reals^{\memlen}$ are zero.  This signal is observed through measurements $\nodevec\in\reals^{\nodenum}$ taken via a compressive (e.g., $\nodenum \ll \memlen$) linear operator:
\begin{gather}
	\nodevec = \ripmat\sigvec + \bm{\epsilon}. \label{eqn:CSlinop}
\end{gather}
Sparse coefficients representing the signal are recovered by solving the convex optimization
\begin{gather}
	\widehat{\bm{a}} = \arg\min_{\bm{a}}\norm{\bm{a}}_1 \quad \mbox{such that} \quad \norm{\nodevec - \ripmat\wavmat\bm{a}}_2\leq \norm{\bm{\epsilon}}_2, \label{eqn:l1min}
\end{gather}
where $\norm{\bm{\epsilon}}_2$ is the magnitude of the measurement noise.  

There is substantial evidence from the signal processing and computational neuroscience communities that many natural signals are sparse in an appropriate basis~\citep{OLS:1996,ELA:2008}.
The recovery problem above requires that the system knows the sparsity basis \wavmat to perform the recovery, which neural systems may not know a priori.  We note that recent work has shown that appropriate sparsity bases can be learned from example data~\citep{OLS:1996}, even in the case where the system only observes the inputs through compressed measurements~\citep{isely2011deciphering}.
While the analysis doesn't depend on the exact method for solving the optimization in equation~\eqref{eqn:l1min}, we also note that this type of optimization can be solved in biologically plausible network architectures (e.g.,~\citep{ROZ:2008,RHE:2007,hu2012network,BAL:2012,BAL:2013,shapero2012jetcas}).    

The most common sufficient condition in CS for stable recovery is known as the Restricted Isometry Property (RIP)~\citep{CAN:2006}.  Formally, we say that RIP-(2\spnum, \ripcd) holds for \ripmat in the basis \wavmat if for any vector \sigvec that is 2\spnum-sparse in \wavmat we have that
\begin{gather}
	\ripcc\left(1-\ripcd\right) \leq \norm{\ripmat\sigvec}_2^2/\norm{\sigvec}_2^2 \leq \ripcc\left(1+\ripcd\right) \label{eqn:RIPstate}
\end{gather}
holds for constants $\ripcc>0$ and $0<\ripcd<1$.  Said another way, the RIP guarantees that all pairs of vectors that are \spnum-sparse in \wavmat have their distances preserved after projecting through the matrix \ripmat.  This can be seen by observing that for a pair of \spnum-sparse vectors, their difference has at most 2\spnum nonzeros.  In this way, the RIP can be viewed as a type of separation property for sparse signals that is similar in spirit to the separation properties used in previous studies of network STM~\citep{JAE:2001,SOM:2004,MAA:2002,hermans2010}.

When $\ripmat$ satisfies the RIP-(2\spnum, \ripcd) in the basis \wavmat with `reasonable' \ripcd (e.g. $\ripcd\leq\sqrt{2}-1$) and the signal estimate is $\widehat{\sigvec}=\wavmat\widehat{\bm{a}}$, canonical results
establish the following bound on signal recovery error:
\begin{gather}
	\norm{\sigvec - \widehat{\sigvec}}_2 \leq \alpha \norm{\bm{\epsilon}}_2 + \beta\frac{\norm{\wavmat^T \left(\sigvec - \sigvec_{\spnum}\right)}_1}{\sqrt{\spnum}},  		\label{eqn:DecRec}
\end{gather}
where $\alpha$ and $\beta$ are constants and $\sigvec_{\spnum}$ is the best \spnum-term approximation to \sigvec in the basis $\wavmat$ (i.e., using the \spnum largest coefficients in $\bm{a}$)~\citep{CAN:2006c}. 
Equation~\eqref{eqn:DecRec} shows that signal recovery error is determined by the magnitude of the measurement noise and sparsity of the signal.  In the case that the signal is exactly \spnum-sparse and there is no measurement noise, this bound guarantees perfect signal recovery.  

While the guarantees above are deterministic and non-asymptotic, the canonical CS results state that measurement matrices generated randomly from ``nice'' independent distributions (e.g., Gaussian, Bernoulli) can satisfy RIP with high probability when $\nodenum=O(\spnum \log \memlen)$~\citep{RAH:2010}.  For example, random Gaussian measurement matrices (perhaps the most highly used construction in CS) satisfy the RIP condition for any sparsity basis with probability $1-O(1/\memlen)$ when $\nodenum \geq C\ripcd^{-2}\spnum\log\left(\memlen\right)$.
This extremely favorable scaling law (i.e., linear in the sparsity level) for random Gaussian matrices is in part due to the fact that Gaussian matrices have many degrees of freedom, resulting in \nodenum statistically independent observations of the signal. 
In many practical examples, there exists a high degree of structure in \ripmat that causes the measurements to be correlated.  Structured measurement matrices with correlations between the measurements have been recently studied due to their computational advantages.  While these matrices can still satisfy the RIP, they typically require more measurement to reconstruct a signal with the same fidelity and the performance may change depending on the sparsity basis (i.e., they are no longer ``universal'' because they don't perform equally well for all sparsity bases). One example which arises often in the signal processing community is the case of random circulant matrices~\citep{krahmer2012suprema}, where the number of measurements needed to assure that the RIP holds with high probability for temporally sparse signals (i.e., \wavmat is the identity) increases to $\nodenum \geq C\ripcd^{-2}\spnum\log^{4}\left( \memlen \right)$.
Other structured systems analyzed in the literature include Toeplitz matrices~\citep{haupt2010toeplitz}, partial circulant matrices~\citep{krahmer2012suprema}, block diagonal matrices~\citep{EFT:2012,park2011concentration}, subsampled unitary matrices~\citep{bajwa2009urip}, and randomly subsampled Fourier matrices~\citep{VER:2008}. 
These types of results are used to demonstrate that signal recovery is possible with highly undersampled measurements, where the number of measurements scales linearly with the ``information level'' of the signal (i.e., the number of non-zero coefficients) and only logarithmically with the ambient dimension.  

\section{STM Capacity using the RIP}

\subsection{Network Dynamics as Compressed Sensing}
\label{sec:NetworkDynamics}

We consider the same discrete-time ESN model used in previous studies~\citep{JAE:2001,SOM:2008,GAN:2010,SOM:2004}:
\begin{gather}
	\nodevect{\tvar} = f\left(\conmat\nodevect{\tvar-1} + \ffvec\sigvart{\tvar} + \widetilde{\bm{\epsilon}}[\tvar]\right), \label{eqn:sysEvolve} 
\end{gather}
where $\nodevect{\tvar}\in\reals^\nodenum$ is the network state at time $\tvar$, \conmat is the ($\nodenum \times \nodenum$) recurrent (feedback) connectivity matrix, $\sigvart{\tvar}\in\reals$ is the input sequence at time $\tvar$, $\ffvec$ is the ($\nodenum \times 1$) projection of the input into the network, $\widetilde{\bm{\epsilon}}[\tvar]$ is a potential network noise source, and $f : \reals^{\nodenum} \rightarrow \reals^{\nodenum}$ is a possible pointwise nonlinearity.  
As in previous studies~\citep{JAE:2001,SOM:2004,SOM:2008,GAN:2010}, this paper will consider the STM capacity of a linear network (i.e., $f\left(\nodevec\right) = \nodevec$).

The recurrent dynamics of Equation~\eqref{eqn:sysEvolve} can be used to write the network state at time \memlen :
\begin{gather}
	\nodevect{\memlen} = \ripmat\sigvec + {\bm{\epsilon}},	\label{eqn:sumMat} 
\end{gather}
where $\ripmat$ is a $\nodenum \times \memlen$ matrix, the $k^{\mbox{th}}$ column of \ripmat is $\conmat^{k-1}\ffvec$, $\sigvec = [\sigvar[\memlen], \dots, \sigvar[1]]^T$, the initial state of the system is $\nodevect{0} = 0$, and ${\bm{\epsilon}}$ is the node activity not accounted for by the input stimulus (e.g. the sum of network noise terms ${\bm{\epsilon}} = \sum_{k=1}^{\memlen}\conmat^{\memlen-k} \widetilde{\bm{\epsilon}}[k]$).  With this network model, we assume that
the input sequence \sigvec is $\spnum$-sparse in an orthonormal basis \wavmat 
(i.e., there are only $\spnum$ nonzeros in $\bm{a}=\wavmat^T \sigvec$).

\subsection{STM Capacity of Finite-Length Inputs} 
\label{sec:FiniteSTM}

We first consider the STM capacity of a network with finite-length inputs, where a length-\memlen input signal drives a network and the current state of the \nodenum network nodes at time \memlen is used to recover the input history via Equation~\eqref{eqn:l1min}.  If \ripmat derived from the network dynamics satisfies the RIP for the sparsity basis \wavmat, the bounds in Equation~\eqref{eqn:DecRec} establish strong guarantees on recovering \sigvec from the current network states $\nodevect{\memlen}$.  Given the significant structure in \ripmat, it is not immediately clear that any network construction can result in \ripmat satisfying the RIP.   
However, the structure in \ripmat is very regular and in fact only depends on powers of \conmat applied to \ffvec: 
\begin{gather}
	\bm{A} = \left[ \begin{matrix} \bm{z} & | & \bm{W}\bm{z} & | & \bm{W}^2\bm{z} & | & \hdots & | & \bm{W}^{N-1}\bm{z} \end{matrix} \right].  \nonumber
\end{gather}
Writing the eigendecomposition of the recurrent matrix $\conmat = \conmatev\conmated\conmatev^{-1}$, we re-write the measurement matrix as 
\begin{gather}
	\bm{A} = \bm{U}\left[ \begin{matrix} \tilde{\bm{z}} & | & \bm{D}\tilde{\bm{z}} & | & \bm{D}^2\tilde{\bm{z}} & | & \hdots & | & \bm{D}^{N-1}\tilde{\bm{z}} \end{matrix} \right], \nonumber
\end{gather}
where $\tilde{\bm{z}} = \bm{U}^{-1}\bm{z}$. Rearranging, we get
\begin{eqnarray}
	\bm{A} &=& \bm{U}\widetilde{\bm{Z}}\left[ \begin{matrix} \bm{d}^0 & | & \bm{d} & | & \bm{d}^2 & | & \hdots & | & \bm{d}^{N-1} \end{matrix} \right]
	= \bm{U}\widetilde{\bm{Z}} \bm{F} \label{eqn:sumMat2}
\end{eqnarray}
where $\fftmat_{k,l} = d_k^{l-1}$ is the $k^{\mbox{th}}$ eigenvalue of \conmat raised to the $(l-1)^{\mbox{th}}$ power and $\zmat = \mbox{diag}\left(\conmatev^{-1}\ffvec\right)$.

While the  RIP conditioning of \ripmat depends on all of the matrices in the decomposition of Equation~\ref{eqn:sumMat2}, the conditioning of \fftmat is the most challenging because it is the only matrix that is compressive (i.e., not square).  Due to this difficulty, we start by specifying a network  structure for \conmatev and \zmat that preserves the conditioning properties of \fftmat (other network constructions will be discussed in Section~\ref{sec:NetProp}).  Specifically, as in~\citep{SOM:2004,SOM:2008,GAN:2010} we choose \conmat to be a random orthonormal matrix,  assuring that the eigenvector matrix \conmatev has orthonormal columns and preserves the conditioning properties of \fftmat.  Likewise, we choose the feed-forward vector \ffvec to be $\ffvec = \frac{1}{\sqrt{\nodenum}}\conmatev\bm{1}_{\nodenum}$, where $\bm{1}_{\nodenum}$ is a vector of \nodenum ones (the constant $\sqrt{\nodenum}$ simplifies the proofs but has no bearing on the result).  This choice for \ffvec assures that \zmat is the identity matrix scaled by $\sqrt{\nodenum}$ (analogous to~\citep{SOM:2008} where \ffvec is optimized to maximize the SNR in the system).
Finally, we observe that the richest information preservation apparently arises for a real-valued \conmat when its  eigenvalues are complex, distinct in phase, have unit magnitude, and appear in complex conjugate pairs.  

For the above network construction, our main result shows that \ripmat satisfies the RIP in the basis \wavmat (implying the bounds from Equation~\eqref{eqn:DecRec} hold) when the network size scales linearly with the sparsity level of the input.  This result is made precise in the following theorem: 
\begin{thm}
	\label{thm:STMbasic}
	Suppose $\memlen \ge \nodenum$, $\memlen \ge \spnum$ and $\memlen \ge O(1)$.\footnote{The notation $\memlen \ge O(1)$ means that $\memlen \ge C$ for some constant $C$. For clarity, we do not keep track of the constants in our proofs. The interested reader is referred to~\citep{RAH:2010} for specific values of the constants.} 
Let $\conmatev$ be any unitary matrix of eigenvectors (containing complex conjugate pairs) and set $\ffvec = \frac{1}{\sqrt{\nodenum}}\conmatev\bm{1}_{\nodenum}$ so that $\zmat = \mbox{diag}\left(\conmatev^{-1}\ffvec\right) = \frac{1}{\sqrt{\nodenum}} \bm{I}$. For $\nodenum$ an even integer, denote the eigenvalues of \conmat by $\{e^{j w_m}\}_{m = 1}^{\nodenum}$.  
	Let the first $\nodenum/2$ eigenvalues $\left(\{e^{j w_m}\}_{m = 1}^{\nodenum/2}\right)$ be chosen uniformly at random on the complex unit circle (i.e., we chose $\{w_m\}_{m=1}^{\nodenum/2}$ uniformly at random from $[0, 2\pi)$) and the other $\nodenum/2$ eigenvalues as the complex conjugates of these values (i.e., for $\nodenum/2 < m \le \nodenum$, $e^{j w_m} = e^{-j w_{m-\nodenum/2}}$). 
Under these conditions, for a given RIP conditioning $\ripcd < 1$ and failure probability $\eta$, if 
\begin{gather}
	\nodenum \geq C\frac{\spnum}{\ripcd^2}\mu^2\left(\wavmat\right)\log^{4}\left(\memlen\right) \log(\eta^{-1}), \label{eqn:NMrelate}
\end{gather}
for a universal constant $C$, then for any $\sigvec$ that is $\spnum$-sparse (i.e., has no more than \spnum non-zero entries)
\begin{gather}
	\left(1-\ripcd\right) \leq \norm{\ripmat\wavmat\sigvec}_2^2/\norm{\sigvec}_2^2 \leq \left(1+\ripcd\right) \nonumber
\end{gather}
with probability exceeding $1- \eta$. 
\end{thm}

The proof of this statement is given in Appendix~\ref{app:BasicRIP} and follows closely the approach in~\citep{RAH:2010} by generalizing it to both include any basis \wavmat and account for the fact that \conmat is a real-valued matrix.
The quantity $\mu\left(\cdot\right)$ (known as the coherence) captures the largest inner product between the sparsity basis and the Fourier basis, and is calculated as:
\begin{gather}
	\mu\left(\wavmat\right) = \max_{n=1,\hdots,N}\sup_{t\in[0,2\pi]} \left|\sum_{m = 0}^{N-1}\wavmat_{m,n}e^{-jtm} \right|.
\end{gather}
In the result above, the coherence is lower  (therefore the STM capacity is higher) when the sparsity basis is more ``different'' from the Fourier basis.

The main observation of the result above is that STM capacity scales superlinearly with network size.  Indeed, for some values of \spnum and  $\mu\left(\wavmat\right)$ it is possible to have STM capacities much greater than the number of nodes (i.e., $\memlen \gg \nodenum$).  To illustrate the perfect recovery of signal lengths beyond the network size, Figure~\ref{fig:Recfig} shows an example recovery of a single long input sequence.  Specifically, we generate a 100 node random orthogonal connectivity matrix \conmat and generate $\ffvec=\frac{1}{\sqrt{\nodenum}}\conmatev\bm{1}_{\nodenum}$. We then drive the network with an input sequence that is 480 samples long and constructed 
using 24 non-zero coefficients (chosen uniformly at random) of a wavelet basis. The values at the non-zero entries were chosen uniformly in the range [0.5,1.5]. In this example we omit noise so that we can illustrate the noiseless recovery. At the end of the input sequence, the resulting 100 network states are used to solve the optimization problem in Equation~\ref{eqn:l1min} for recovering the input sequence (using the network architecture in~\citep{ROZ:2008}). The recovered sequence, as depicted in Figure~\ref{fig:Recfig}, is identical to the input sequence, clearly indicating that the 100 nodes were able to store the 480 samples of the input sequence (achieving STM capacity higher than the network size).

Directly checking the RIP condition for specific matrices is NP-hard (one would need to check every possible $2\spnum$-sparse signal).  In light of this difficulty in verifying recovery of all possible sparse signals (which the RIP implies), we will explore the qualitative behavior of the RIP bounds above by examining in Figure~\ref{fig:Phasefig} the average recovery relative MSE (rMSE) in simulation for a network with \nodenum nodes when recovering input sequences of length \memlen with varying sparsity bases.  Figure~\ref{fig:Phasefig} uses a plotting style similar to the Donoho-Tanner phase transition diagrams~\citep{DON:2005} where the average recovery rMSE is shown for each pair of variables under noisy conditions. While the traditional Donoho-Tanner phase transitions plot noiseless recovery performance to observe the threshold between perfect and imperfect recovery, here we also add noise to illustrate the stability of the recovery guarantees.
The noise is generated as random additive Gaussian noise at the input ($\widetilde{\bm{\epsilon}}$ in Equation~\eqref{eqn:sysEvolve}) to the system with zero mean and variance such that the total noise in the system ($\epsilon$ in Equation~\eqref{eqn:sumMat}) has a norm of approximately 0.01.  
To demonstrate the behavior of the system, the phase diagrams in Figure~\ref{fig:Phasefig} sweep the ratio of measurements to the total signal length (\nodenum/\memlen) and the ratio of the signal sparsity to the number of measurements (\spnum/\nodenum). Thus at the upper left hand corner, the system is recovering a dense signal from almost no measurements (which should almost certainly yield poor results) and at the right hand edge of the plots the system is recovering a signal from a full set of measurements (enough to recover the signal well for all sparsity ranges). 
We generate ten random ESNs for each combination of ratios (\nodenum/\memlen, \spnum/\nodenum). The simulated networks are driven with input sequences that are  sparse in one of four different bases (Canonical, Daubechies-10 wavelet, Symlet-3 wavelet and DCT) which have varying coherence with the Fourier basis.  We use the node values at the end of the sequence to recover the inputs.\footnote{For computational efficiency, we use the TFOCS software package~\citep{tfocs2011} to solve the optimization problem in Equation~\eqref{eqn:l1min} for these simulations.} 

In each plot of Figure~\ref{fig:Phasefig}, the dashed line denotes the boundary where the system is able to essentially perform perfect recovery (recovery error $\leq$ 1\%) up to the noise floor.  Note that the area under this line (the white area in the plot) denotes the region where the system is leveraging the sparse structure of the input to get capacities of $\memlen > \nodenum$.  We also observe that the dependence of the RIP bound on the coherence with the Fourier basis is clearly shown qualitatively in these plots, with the DCT sparsity basis showing much worse performance than the other bases.

\begin{figure}
	\begin{center}
		\includegraphics[width=.7\textwidth]{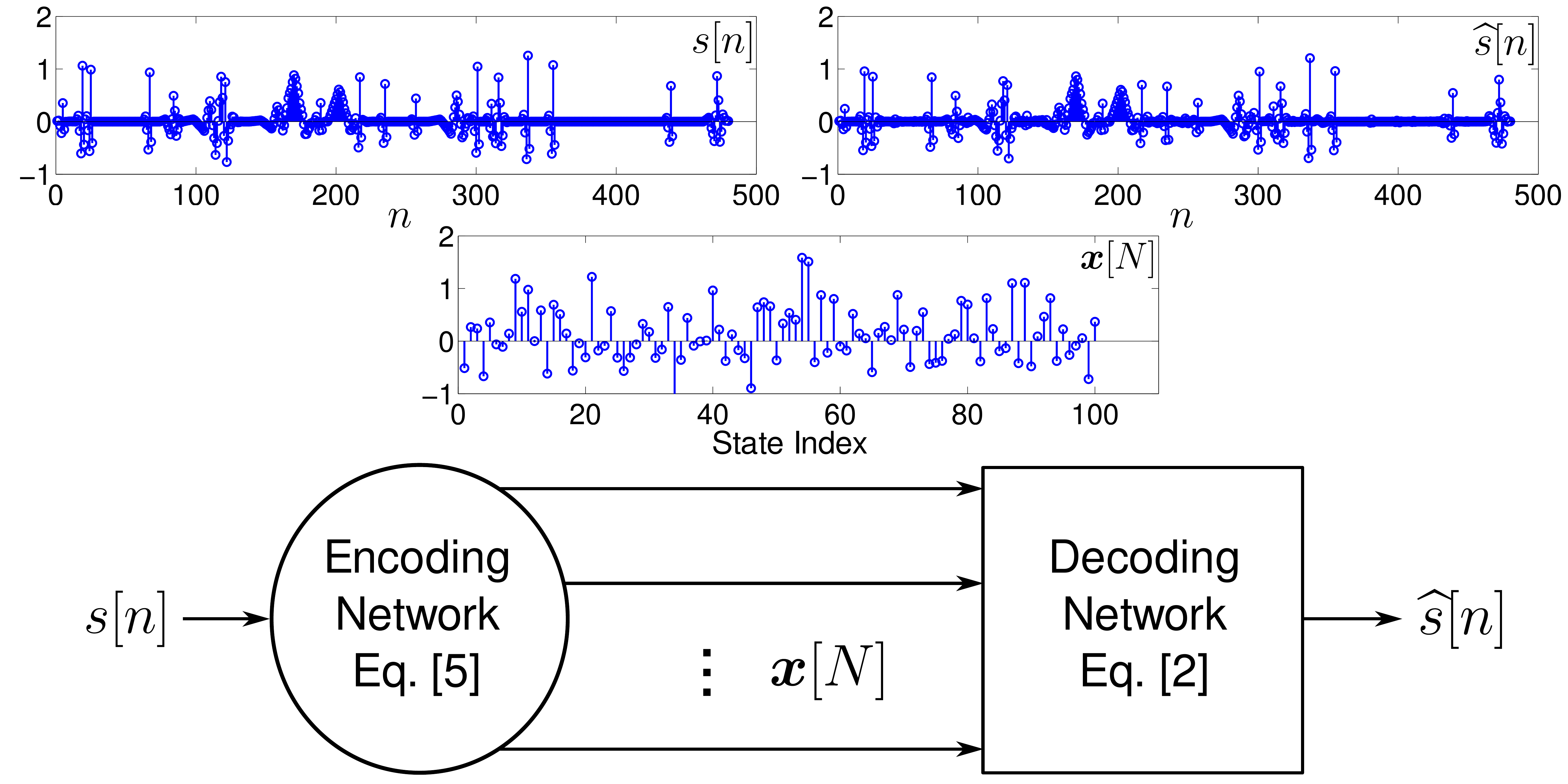}
	\end{center}
	\caption{A length 480 stimulus pattern (left plot) that is sparse in a wavelet basis drives the encoding network defined by a random orthogonal matrix \conmat and a feed-forward vector \ffvec. The 100 node values (center plot) are then used to recover the full stimulus pattern (right plot) using a decoding network which solves Equation~\eqref{eqn:l1min}.}
	\label{fig:Recfig}
\end{figure}

\begin{figure}
	\begin{center}
		\includegraphics[width=\textwidth]{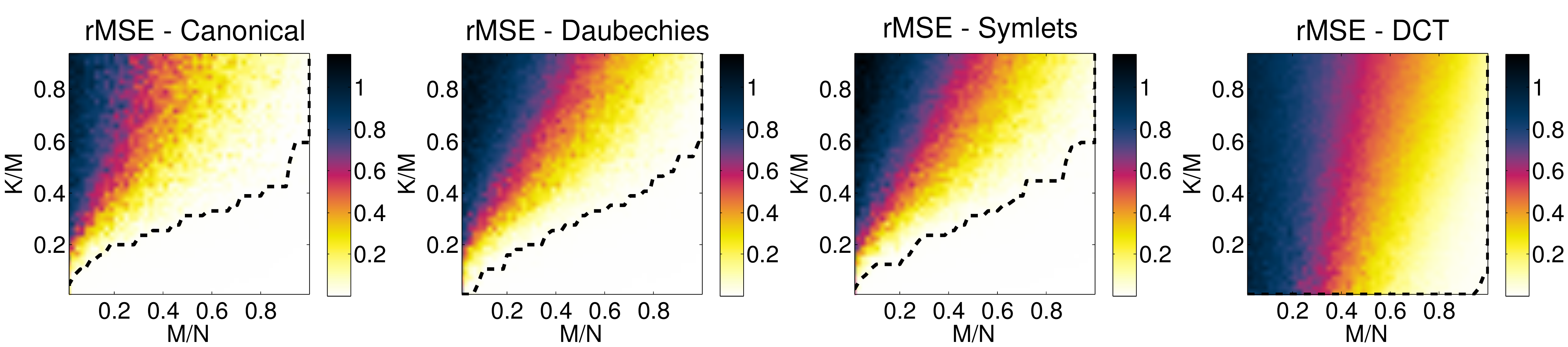}
	\end{center}
	\caption{Random orthogonal networks can have a STM capacity that exceeds the number of nodes. These plots depict the recovery relative mean square error (rMSE) for length-1000 input sequences  from \nodenum  network nodes where the input sequences are $\spnum$-sparse. Each figure depicts recovery for a given set of ratios $\nodenum/\memlen$ and $\spnum/\nodenum$. Recovery is near perfect (rMSE $\leq$ 1\%; denoted by the dotted line) for large areas of each plot (to the left of the \memlen = \nodenum boundary at the right of each plot) for sequences sparse in the canonical basis or various wavelet basis (shown here are 4 level decompositions in Symlet-3 wavelets and Daubechies-10 wavelets). For bases more coherent with the Fourier basis (e.g., discrete cosine transform-DCT), recovery performance above \memlen = \nodenum can suffer significantly. All the recovery here was done for noise such that $\|\bm{\epsilon}\|_2 \approx 0.01$.}
	\label{fig:Phasefig}
\end{figure}

\subsection{STM Capacity of Infinite-Length Inputs} 
\label{sub:InfiniteSTM}

After establishing the perfect recovery bounds for finite-length inputs in the previous section, we turn here to the more interesting case of a network that has received an input beyond its STM capacity (perhaps infinitely long).   In contrast to the finite-length input case where favorable constructions for \conmat used random unit-norm eigenvalues, this construction  would be unstable for infinitely long inputs.  In this case, we take \conmat to have all eigenvalue magnitudes equal to $\decvar<1$ to ensure stability.  The matrix constructions we consider in this section are otherwise identical to that described in the previous section.

In this scenario, the recurrent application of \conmat in the system dynamics assures that each input perturbation will decay steadily until it has zero effect on the network state.  While good for system stability, this decay means that each input will slowly recede into the past until  the network activity contains no useable memory of the event.  In other words, \emph{any} network with this decay can only hope to recover a proxy signal that accounts for the decay in the signal representation induced by the forgetting factor \decvar.  Specifically, we define this proxy signal to be $\decmat\sigvec$, where $\decmat = \mbox{diag}\left(\left[1, \decvar, \decvar^2, \hdots \right]\right)$.  Previous work~\citep{SOM:2008,JAE:2001,SOM:2004} has characterized recoverability by using statistical arguments to quantify the correlation of the node values to each past input perturbation.  In contrast, our approach is to provide recovery bounds on the rMSE for a network attempting to recover the \memlen past samples of $\decmat\sigvec$, which corresponds to the weighted length-\memlen history of \sigvec.  Note that in contrast to the previous section where we established the length of the input that can be perfectly recovered, the amount of time we attempt to recall (\memlen) is now a parameter that can be varied.

Our technical approach to this problem comes from observing that activity due to inputs older than \memlen acts as interference when recovering more recent inputs.  
In other words, we can group older terms (i.e., from farther back than \memlen time samples ago) with the noise term, resulting again in \ripmat being an \nodenum by \memlen linear operation that can satisfy RIP for length-\memlen inputs.  In this case, after choosing the length of the memory to recover, the guarantees in Equation~\eqref{eqn:DecRec} hold when considering every input older than \memlen as contributing to the ``noise'' part of the bound.  


Specifically, in the noiseless case where \sigvec is sparse in the canonical basis ($\mu\left(\bm{I}\right) = 1$) with a maximum signal value \sigmax, we can bound the first term of Equation~\eqref{eqn:DecRec} using a geometric sum that depends on \memlen, \spnum and \decvar.  
For a given scenario (i.e., a choice of \decvar, \spnum and the RIP conditioning of \ripmat), a network can support signal recovery up to a certain sparsity level \spnumf, given by: 
\begin{gather}
        \spnumf =  \frac{\nodenum\ripcd^2}{\ripcc\log^\ripex\left( \memlen \right)}, \label{eqn:decay_bound2}
\end{gather}
where $\ripex$ is a scaling constant (e.g., $\ripex=4$ using the present techniques, but $\ripex=1$ is conjectured~\citep{VER:2008}).
We can also bound the second term of Equation~\eqref{eqn:DecRec} by the sum of the energy in the past \memlen perturbations that are beyond this sparsity level \spnumf.  Together these terms yield the bound on the recovery of the proxy signal:
\begin{eqnarray}
        \norm{\decmat\sigvec - \decmat\sigvest}_2 & \leq & \beta\sigmax\norm{\conmatev}_2\left(\frac{\decvar^{\memlen}}{1-\decvar}\right)  \nonumber \\ 
	&  + & \frac{\beta\sigmax}{\sqrt{\min\left[\spnumf, \spnum\right]}}\left(\frac{\decvar^{\min\left[\spnumf, \spnum\right]} - \decvar^{\spnum}}{1-\decvar}\right)  \label{eqn:decay_bound} \\
	& + & \alpha \epsilon_{\max{}} \norm{\conmatev}_2 \left| \frac{\decvar}{1 - \decvar} \right|. \nonumber
\end{eqnarray}
The derivation of the first two terms in the above bound is detailed in Appendix~\ref{app:HistErr}, and the final term is simply the accumulated noise, which should have bounded norm due to the exponential decay of the eigenvalues of \conmat. 

Intuitively, we see that this approach implies the presence of an optimal value for the recovery length \memlen. For example, choosing \memlen too small means that there is useful signal information in the network that the system is not attempting to recover, resulting in omission errors (i.e., an increase in the first term of Equation~\eqref{eqn:DecRec} by counting too much signal as noise).  On the other hand, choosing \memlen too large means that the system is  encountering recall errors by trying to recover inputs with little or no residual information remaining in the network activity (i.e., an increase in the second term of Equation~\eqref{eqn:DecRec} from making the signal approximation worse by using the same number of nodes for a longer signal length).

The intuitive argument above can be made precise in the sense that the bound in Equation~\eqref{eqn:decay_bound} does have at least one local minimum for some value of $0<\memlen < \infty$.  First, we note that the noise term (i.e., the third term on the right side of Equation~\eqref{eqn:decay_bound}) does not depend on \memlen (the choice in origin does not change the infinite summation), implying that the optimal recovery length only depends on the first two terms.  We also note the important fact that \spnumf is non-negative and monotonically decreasing with increasing \memlen .  It is straightforward to observe that the bound in equation Equation~\eqref{eqn:decay_bound} tends to infinity as \memlen increases (due to the presence of \spnumf in the denominator of the second term).  Furthermore, for small values of \memlen, the second term in Equation~\eqref{eqn:decay_bound} is zero (due to $\spnumf > \spnum$), and the first term is monotonically decreasing with \memlen .  Taken together, since the function is continuous in \memlen, has negative slope for small \memlen and tends to infinity for large \memlen, we can conclude that it must have at least one local minima in the range $0<\memlen < \infty$.  This result predicts that there is (at least one) optimal value for the recovery length \memlen .

The prediction of an optimal recovery length above is based on the fact that the error bound in Equation~\eqref{eqn:decay_bound}), and it is possible that the error itself will not actually show this behavior (since the bound may not be tight in all cases).  To test the qualitative intuition from Equation~\eqref{eqn:decay_bound}, we simulate recovery of input lengths and show the results in Figure~\ref{fig:LIfig}.  Specifically, we generate 50 ESNs with 500 nodes and a decay rate of \decvar=0.999. The input signals are length-8000 sequences that have 400 nonzeros whose locations are chosen uniformly at random and whose amplitudes are chosen from a Gaussian distribution (zero mean and unit variance). After presenting the full 8000 samples of the input signal to the network, we use the network states to recover the input history with varying lengths and compared the resulting MSE to the bound in Equation~\eqref{eqn:decay_bound}. Note that while the theoretical bound may not be tight for large signal lengths, the recovery MSE matches the qualitative behavior of the bound by achieving a minimum value at $\memlen>\nodenum$. 

\begin{figure}
	\begin{center}
		\includegraphics[width=.9\textwidth]{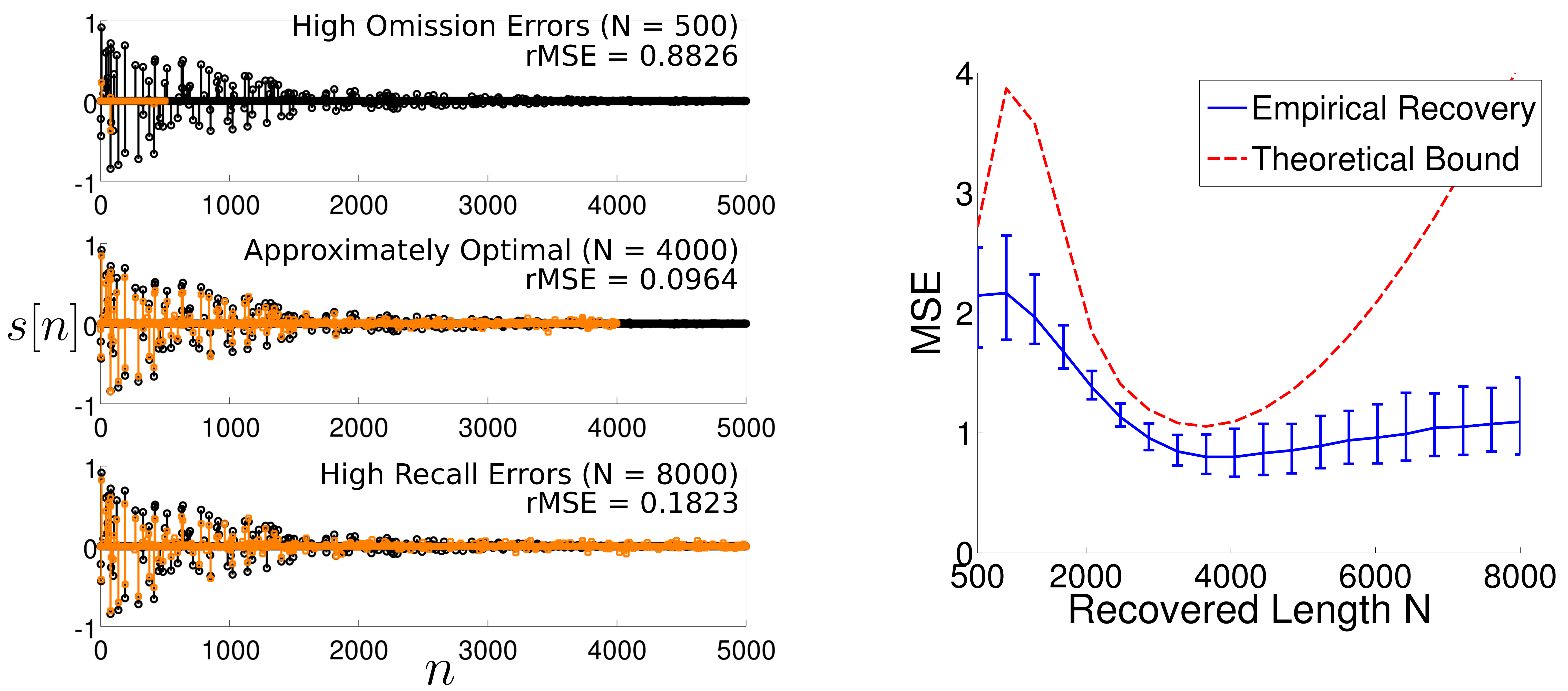}
	\end{center}
	\caption{The theoretical bound on the recovery error for the past \memlen perturbations to a network of size \nodenum has a minimum value at some optimal recovery length. This optimal value depends on the network size, the sparsity \spnum, the decay rate \decvar, and the RIP conditioning of \ripmat. Shown on the right is a simulation depicting the MSE for both the theoretical bound (red dashed line) and an empirical recovery for varying recovery lengths \memlen. In this simulation \spnum = 400, \decvar = 0.999, \nodenum = 500. The error bars for the empirical curve show the maximum and minimum MSE. On the left we show recovery (in orange) of a length-8000 decayed signal (in black) when recovering the past 500 (top), 4000 (middle), and 8000 (bottom) most recent perturbations. As expected, at \memlen = 4000 (approximately optimal) the recovery has the highest accuracy.}
	\label{fig:LIfig}
\end{figure}

\section{Other Network Constructions}
\label{sec:NetProp}

\subsection{Alternate Orthogonal Constructions}

Our results in the previous section focus on the case where \conmat is orthogonal and \ffvec  projects the signal evenly into all eigenvectors of \conmat. When either \conmat or \ffvec deviate from this structure the STM capacity of the network apparently decreases.  In this section we revisit those specifications, considering alternate network structures allowed under these assumptions as well as the consequences of deviating from these assumptions in favor of other structural advantages for a system (e.g., wire length, etc.).

To begin, we consider the assumption of orthogonal network connectivity, where the eigenvalues have constant magnitude and the eigenvectors are orthonormal. Constructed in this way, \conmatev exactly preserves the conditioning of $\zmat\fftmat$. While this construction may seem restrictive, orthogonal matrices are relatively simple to generate and encompass a number of distinct cases. For small networks, selecting the eigenvalues uniformly at random from the unit circle (and including their complex conjugates to ensure real connectivity weights) and choosing an orthonormal set of complex conjugate eigenvectors creates precisely these optimal properties. For larger matrices, the connectivity matrix can instead be constructed directly by choosing \conmat at random and orthogonalizing the columns. Previous results on random matrices~\citep{DIA:1994} guarantee that as the size of \conmat increases, the eigenvalue probability density approaches the uniform distribution as desired. 
Some recent work in STM capacity demonstrates an alternate method by which orthogonal matrices can be constructed while constraining the total connectivity of the network~\citep{strauss2012design}. This method iteratively applies rotation matrices to obtain orthogonal matrices with varying degrees of connectivity.
We note here that one special case of connectivity matrices not well-suited to the STM task, even when made orthogonal, are symmetric networks, where the strictly real-valued eigenvalues generates poor RIP conditioning for \fftmat.

While simple to generate in principle, the matrix constructions discussed above are generally densely connected and may be impractical for many systems. 
However, many other special network topologies that may be more biophysically realistic (i.e., block diagonal connectivity matrices and small-world\footnote{Small-world structures are typically taken to be networks where small groups of neurons are densely connected amongst themselves, yet sparse connections to other groups reduces the maximum distance between any two nodes.} networks~\citep{MON:2008}) can be constructed so that \conmat still has orthonormal columns.  For example, consider the case of a block diagonal connection matrix (illustrated in Figure~\ref{fig:OtopFig}), where many unconnected networks of at least two nodes each are driven by the same input stimulus and evolve separately.  Such a structure lends itself to a modular framework, where more of these subnetworks can be recruited to recover input stimuli further in the past.  In this case, each block can be created independently as above and pieced together. 
The columns of the block diagonal matrix will still have unit norm and will be both orthogonal to vectors within its own block (since each of the diagonal sub-matrices are orthonormal) and orthogonal to all columns in other blocks (since there is no overlap in the non-zero indices). 

\begin{figure}
	\begin{center}
		\includegraphics[width=4.5in]{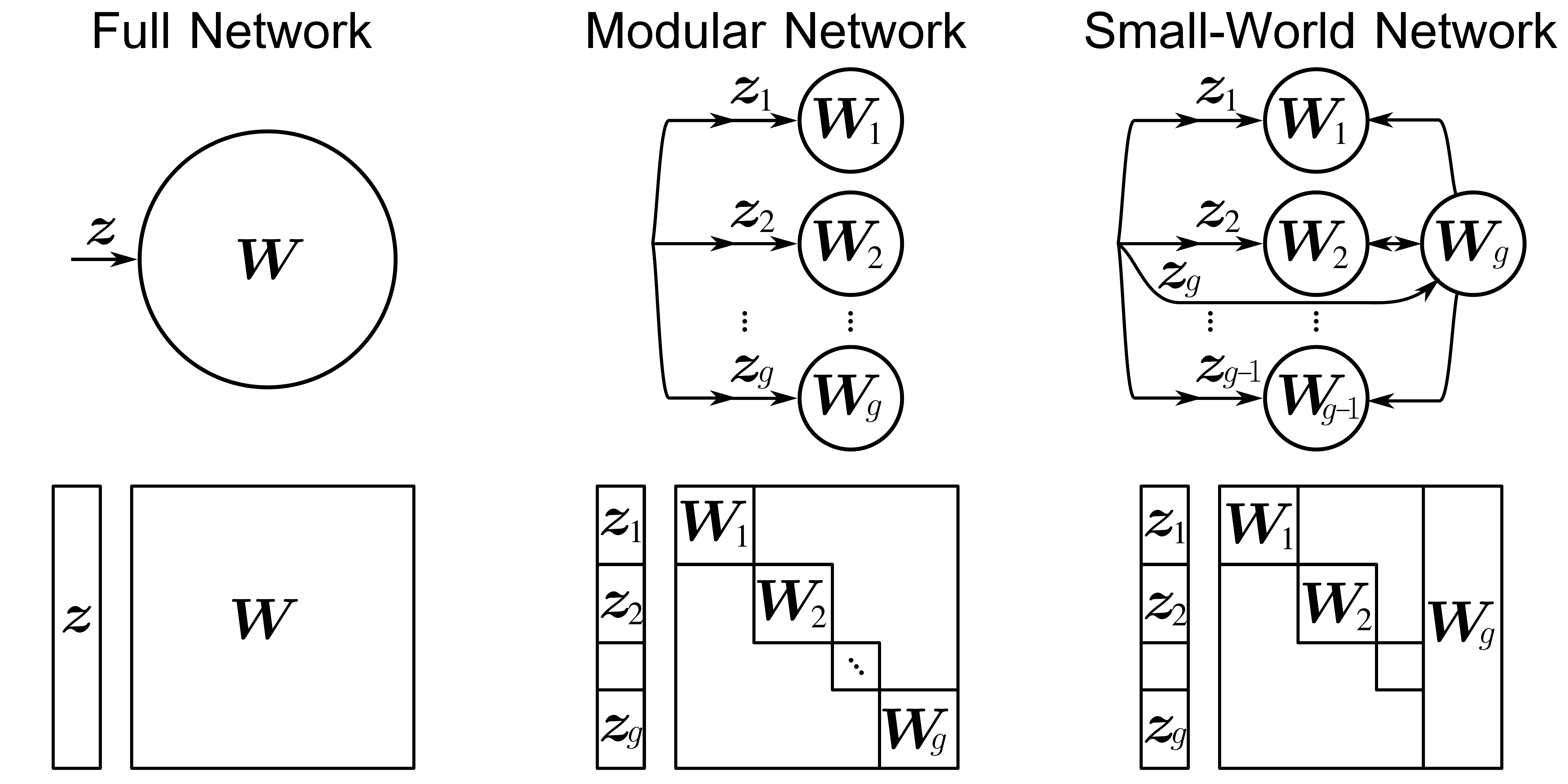}
	\end{center}
	\caption{Possible network topologies which have orthogonal connectivity matrices. In the general case, all nodes are connected via non-symmetric connections. Modular topologies can still be orthogonal if each block is itself orthogonal. Small world topologies may also have orthogonal connectivity, especially when a few nodes are completely connected to a series of otherwise disjoint nodes.}
	\label{fig:OtopFig}
\end{figure}

Similarly, a small-world topology can be achieved by taking a few of the nodes in every group of the block diagonal case and allowing connections to all other neurons (either unidirectional or bidirectional connections). To construct such a matrix, a block diagonal orthogonal matrix can be taken, a number of columns can be removed and replaced with full columns, and the resulting columns can be made orthonormal with respect to the remaining block-diagonal columns. In these cases, the same eigenvalue distribution and eigenvector properties hold as the fully connected case, resulting in the same RIP guarantees (and therefore the same recovery guarantees) demonstrated earlier.   We note that this is only one approach to constructing a network with favorable STM capacity and not all networks with small-world properties will perform well.

Additionally, we note that as opposed to networks analyzed in prior work (in particular the work in~\citep{latham2013} demonstrating that random networks with high connectivity have short STM), the average connectivity does not play a dominant role in our analysis.  Specifically, it has been observed in spiking networks that higher network connectivity can reduce the STM capacity so that is scales only with $\log(\nodenum)$~\citep{latham2013}). However, in our ESN analysis, networks can have low connectivity (e.g. 2x2 block-diagonal matrices - the extreme case of the block diagonal structure described above) or high connectivity (e.g. fully connected networks) and have the same performance. 

\subsection{Suboptimal Network Constructions}

Finally, we can also analyze some variations to the network structure assumed in this paper to see how much performance decreases.  First, instead of the deterministic construction for \ffvec discussed in the earlier sections, there has also been interest in choosing \ffvec as i.i.d.\ random Gaussian values~\citep{SOM:2008,GAN:2010}.  In this case, it is also possible to show that \ripmat satisfies the RIP (with respect to the basis \wavmat and with the same RIP conditioning \ripcd as before) by paying an extra $\log(\memlen)$ penalty in the number of measurements. 
Specifically, we have also established the following theorem:

\begin{thm}
	\label{thm:STMwithZ}
	Suppose $\memlen \ge \nodenum$, $\memlen \ge \spnum$ and $\memlen \ge O(1)$. 
	Let $\conmatev$ be any unitary matrix of eigenvectors (containing complex conjugate pairs) and the entries of \ffvec be i.i.d.\ zero-mean Gaussian random variables with variance $\frac{1}{\nodenum}$.
	For $\nodenum$ an even integer, denote the eigenvalues of \conmat by $\{e^{j w_m}\}_{m = 1}^{\nodenum}$.  Let the first $\nodenum/2$ eigenvalues ($\{e^{j w_m}\}_{m = 1}^{\nodenum/2}$) be chosen uniformly at random on the complex unit circle (i.e., we chose $\{w_m\}_{m=1}^{\nodenum/2}$ uniformly at random from $[0, 2\pi)$) and the other $\nodenum/2$ eigenvalues as the complex conjugates of these values.   
Then, for a given RIP conditioning $\ripcd$ and failure probability $\memlen^{-\log^4 \memlen} \le \eta \le \frac{1}{e}$, if 
\begin{gather}
	\nodenum \geq C\frac{\spnum}{\ripcd^2}\mu^2\left(\wavmat\right)\log^{5}\left(\memlen\right) \log(\eta^{-1}), \label{eqn:NMrelateZ}
\end{gather}
$\ripmat$ satisfies RIP-$(\spnum,\ripcd)$ with probability exceeding $1-\eta$ for a universal constant $C$.
\end{thm}

The proof of this theorem can be found in Appendix~\ref{app:RIPz}. The additional log factor in the bound in~\eqref{eqn:NMrelateZ} reflects that a random feed-forward vector may not optimally spread the input energy over the different eigen-directions of the system. Thus, some nodes may see less energy than others, making them slightly less informative. Note that while this construction does perform worse that the optimal constructions from Theorem~\ref{thm:STMbasic}, the STM capacity is still very favorable (i.e.,  a linear scaling in the sparsity level and logarithmic scaling in the signal length).

Second, instead of orthogonal connectivity matrices, there has also been interest in network constructions involving non-orthogonal connectivity matrices (perhaps for noise reduction purposes~\citep{SOM:2008}).  
When the eigenvalues of \conmat still lie on the  complex unit circle, we can analyze how non-orthogonal matrices affect the RIP results.
In this case, the decomposition in Equation~\eqref{eqn:sumMat2} still holds and Theorem~\ref{thm:STMbasic} still applies to guarantee that \fftmat satisfies the RIP. However, the non-orthogonality changes the conditioning of \conmatev and subsequently the total conditioning of \ripmat .
Specifically the conditioning of \conmatev (the ratio of the maximum and minimum singular values $\sigma_{\max}^2/\sigma_{\min}^2 = \gamma$) will effect the total conditioning of \ripmat. We can use the RIP of \fftmat and the extreme singular values of \conmatev to bound how close $\conmatev\fftmat$ is to an isometry for sparse vectors, both above by
\begin{gather}
	\norm{\conmatev\fftmat\sigvec}_2^2 \leq \sigma_{\max}^2\norm{\fftmat\sigvec}_2^2 \leq \sigma_{\max}^2\ripcc(1+\ripcd)\norm{\sigvec}_2^2,\nonumber
\end{gather}
and below by
\begin{gather}
	\norm{\conmatev\fftmat\sigvec}_2^2 \geq \sigma_{\min}^2\norm{\fftmat\sigvec}_2^2 \geq \sigma_{\min}^2\ripcc(1-\ripcd)\norm{\sigvec}_2^2.\nonumber
\end{gather}
By consolidating these bounds, we find a new RIP statement for the composite matrix
\begin{gather}
	\ripcc'\left(1-\ripcd'\right)\norm{\sigvec}_2^2 \leq \norm{\conmatev\fftmat\sigvec}_2^2 \leq \ripcc'\left(1+\ripcd'\right)\norm{\sigvec}_2^2 \nonumber
\end{gather}
where $\sigma_{\min}^2\ripcc(1-\ripcd)  =  \ripcc'(1-\ripcd')$ and $\sigma_{\max}^2\ripcc(1+\ripcd)  =  \ripcc'(1+\ripcd')$.
These relationships can be used to solve for the new RIP constants:
\begin{eqnarray}
	\ripcd' & = & \frac{\frac{\gamma - 1}{\gamma + 1} + \ripcd}{1 + \ripcd\frac{\gamma - 1}{\gamma + 1}} \nonumber  \\
	\ripcc' & = & \frac{1}{2}\ripcc\left(\sigma_{\max}^2 + \sigma_{\min}^2 + \ripcd(\sigma_{\max}^2 - \sigma_{\min}^2)\right) \nonumber
\end{eqnarray}
These expressions demonstrate that as the conditioning of \conmatev improves (i.e. $\gamma \rightarrow 1$), the RIP conditioning does not change from the optimal case of an orthogonal network ($\ripcd' = \ripcd$). However, as the conditioning of \conmatev gets worse and $\gamma$ grows, the constants associated with the RIP statement also get worse (implying more measurements are likely required to guarantee the same recovery performance).

The above analysis primarily concerns itself with constructions where the eigenvalues of \conmat are still unit norm, however \conmatev is not orthogonal. Generally, when the eigenvalues of \conmat differ from unity and are not all of equal magnitude, the current approach becomes intractable. In one case, however, there are theoretical guarantees: when \conmat is rank deficient. If \conmat only has $\widetilde{\nodenum}$ unit-norm eigenvalues, and the remaining $\nodenum - \widetilde{\nodenum}$ eigenvalues are zero, then the resulting matrix \ripmat is composed the same way, except that the bottom $\nodenum - \widetilde{\nodenum}$ rows are all zero. This means that the effective measurements only depend on an $\widetilde{\nodenum}\times\memlen$ subsampled DTFT
\begin{eqnarray}
	\nodevect{\memlen} & = & \conmatev\zmat\fftmat\sigvec + \bm{\epsilon} \nonumber \\
	                   & = & \conmatev\zmat\left[\begin{matrix}\widetilde{\fftmat} \\ \bm{0}_{\nodenum-\widetilde{\nodenum},\memlen}\end{matrix} \right]\sigvec + \bm{\epsilon} \nonumber \\
				   & = & \conmatev\zmat_{1:\widetilde{M}}\widetilde{\fftmat}\sigvec + \bm{\epsilon} \nonumber
\end{eqnarray}
where $\widetilde{\fftmat}$ is matrix consisting of the non-zero rows of \fftmat. In this case we can choose any $\widetilde{\nodenum}$ of the nodes and the previous theorems will all hold, replacing the true number of nodes \nodenum with the effective number of nodes $\widetilde{\nodenum}$.

\section{Discussion}
\label{sec:conc}

We have seen that the tools of the CS literature can provide a way to quantify the STM capacity in linear networks using rigorous non-asymptotic recovery error bounds.  Of particular note is that this approach leverages the non-Gaussianity of the input statistics to show STM capacities that are superlinear in the size of the network and depend linearly on the sparsity level of the input.  This work provides a concrete theoretical understanding for the approach conjectured in~\citep{GAN:2010} along with a generalization to arbitrary sparsity bases and infinitely long input sequences.  This analysis also predicts that there exists an optimal recovery length that balances omission errors and recall mistakes.  

In contrast to previous work on ESNs that leverage nonlinear network computations for computational power~\citep{JAE:2004}, the present work uses a linear network and nonlinear computations for signal recovery. Despite the nonlinearity of the recovery process, the fundamental results of the CS literature also guarantee that the recovery process is stable and robust.  For example, with access to only a subset of nodes (due to failures or communication constraints), signal recovery generally degrades gracefully by still achieving the best possible approximation of the signal using fewer coefficients.  Beyond signal recovery, we also note that the RIP can guarantee performance on many tasks (e.g. detection, classification, etc.) performed directly on the network states~\citep{DAV:2010}. 
Finally, we note that while this work only addresses the case where a single input is fed to the network, there may be networks of interest that have a number of input streams all feeding into the same network (with different feed-forward vectors). We believe that the same tools utilized here can be used in the multi-input case, since the overall network state is still a linear function of the inputs.


\subsection*{Acknowledgments}
The authors are grateful to J. Romberg for valuable discussions related to this work. This work was partially supported by NSF grant CCF-0905346 and DSO National Laboratories, Singapore. 

\section{Appendix}

\subsection{Proof of RIP}
\label{app:BasicRIP}

In this appendix, we show that the matrix $\ripmat = \conmatev\zmat\fftmat$ satisfies the RIP under the conditions stated in Equation~\eqref{eqn:NMrelate} of the main text in order to prove Theorem~\ref{thm:STMbasic}.
We note that~\citep{RAH:2010} shows that for the canonical basis ($\wavmat = \bm{I}$), the bounds for $\nodenum$ can be tightened to $\nodenum \ge \max\left\{C \frac{\spnum}{\ripcd^2} \log^4 \memlen,\; C' \frac{\spnum}{\ripcd^2} \log \eta^{-1} \right\}$ using a more complex proof technique than we will employ here.
For $\eta = \frac{1}{\memlen}$, the result in~\citep{RAH:2010} represents an improvement of several $\log (\memlen)$ factors when restricted to only the canonical basis for \wavmat. We also note that the scaling constant \ripcc found in the general RIP definition of Equation~\eqref{eqn:RIPstate} of the main text is unity due to the $\sqrt{\nodenum}$ scaling of \ffvec. 

While the proof of Theorem~\ref{thm:STMbasic} is fairly technical, the procedure follows very closely the proof of Theorem 8.1 from~\citep{RAH:2010} on subsampled discrete time Fourier transform (DTFT) matrices.  While the basic approach is the same, the novelty in our presentation is the incorporation of the sparsity basis $\wavmat$ and considerations for a real-valued connectivity matrix $\conmat$. 

Before beginning the proof of this theorem, we note that because
$\conmatev$ is assumed unitary, $\|\ripmat \wavmat \sigvec\|_2 = \|\zmat \fftmat \wavmat \sigvec\|_2$ for any signal $\sigvec$. 
Thus, it suffices to establish the conditioning properties of the matrix $\widehat{\ripmat} := \zmat \fftmat \wavmat$.  For the upcoming proof, it will be useful to write this matrix as a sum of 
rank-1 operators.  The specific rank-1 operator that will be useful for our purposes is $X_l X_l^H$ with $X_l^H := \fftmat_l^H \wavmat$, the conjugate of the $l$-th row of $\fftmat \wavmat$, where $\fftmat_l^H := \left[1,\; e^{jw_l},\; \cdots,\; e^{j w_l (N-1)} \right] \in \comps^{\memlen}$ is the conjugated $l$-th row of $\fftmat$. 
Because of the way the ``frequencies'' $\{w_m\}$ are chosen, for any $l > \frac{M}{2}$, $X_l = X_{l - \frac{M}{2}}^*$. 
The $l$-th row of $\widehat{\ripmat}$ is $\ffvecu_l X_l^H$ where $\ffvecu_l$ is the $l$-th diagonal entry of the diagonal matrix $\zmat$, meaning that we can use the sum of rank-1 operators to write the decomposition $\widehat{\ripmat}^H \widehat{\ripmat} = \sum_{l = 1}^\nodenum |\ffvecu_l|^2 X_l X_l^H$.  If we define the random variable $\bm{B} := \widehat{\ripmat}^H \widehat{\ripmat} - \bm{I}$ and the norm
 $\displaystyle \|\bm{B}\|_\spnum := \sup_{\tiny y\mbox{ is \spnum-sparse}}\frac{\bm{y}^H\bm{B}\bm{y}}{\bm{y}^H\bm{y} }$, we can equivalently say that $\widehat{\ripmat}$ has RIP conditioning $\ripcd$ if
\begin{eqnarray*}
	\left\| \bm{B} \right\|_\spnum := 
	\left\|\widehat{\ripmat}^H \widehat{\ripmat} - \bm{I} \right\|_\spnum =
	\left\|\sum_{l = 1}^\nodenum |\ffvecu_l|^2 X_l X_l^H - \bm{I} \right\|_\spnum \le \ripcd. 
\end{eqnarray*}

To aid in the upcoming proof, we make a few preliminary observations and rewrite the quantities of interest in some useful ways.  First,  because of the correspondences between the summands in $\widehat{\ripmat}^H \widehat{\ripmat}$ (i.e. $X_l = X_{l-M/2}^{*}$), we can rewrite $\widehat{\ripmat}^H \widehat{\ripmat}$ as
\begin{eqnarray*}
	\widehat{\ripmat}^H \widehat{\ripmat} &=& \sum_{l = 1}^{{\nodenum}/{2}} |\ffvecu_l|^2 X_l X_l^H + \sum_{l = 1}^{{\nodenum}/{2}} |\ffvecu_l|^2 \left(X_l X_l^H\right)^{*},
\end{eqnarray*}
making clear the fact that there are only $\frac{\nodenum}{2}$ independent $w_m$'s. 
Under the assumption of Theorem~\ref{thm:STMbasic}, $\ffvecu_l = \frac{1}{\sqrt{\nodenum}}$ for $l = 1, \cdots, \nodenum$. 
Therefore,
\begin{eqnarray*}
	\expect{\sum_{l = 1}^{{\nodenum}/{2}} |\ffvecu_l|^2 X_l X_l^H} 
	= \sum_{l = 1}^{\nodenum/2} |\ffvecu_l|^2 \expect{X_l X_l^H}
	= \sum_{l = 1}^{\nodenum/2} \frac{1}{\nodenum} \wavmat^H \expect{\fftmat_l \fftmat_l^H} \wavmat 
	= \frac{1}{2} \bm{I},
\end{eqnarray*}
where it is straightforward to check that $\expect{\fftmat_l \fftmat_l^H} = \bm{I}$. 
By the same reasoning, we also have $\expect{\sum_{l = 1}^{{\nodenum}/{2}} |\ffvecu_l|^2 \left(X_l X_l^H\right)^{*}} = \frac{1}{2} \bm{I}$. 
This implies that we can rewrite $\bm{B}$ as
\begin{eqnarray*}
	\bm{B} & = & \sum_{l = 1}^\nodenum \left(|\ffvecu_l|^2 {X_l X_l^H} \right) - \bm{I} \\
	& = & \left(\sum_{l = 1}^{\nodenum/2} |\ffvecu_l|^2 {X_l X_l^H} - \frac{1}{2} \bm{I}\right) + 
	\left(\sum_{l = 1}^{\nodenum/2} |\ffvecu_l|^2 \left(X_l X_l^H\right)^{*} - \frac{1}{2} \bm{I}\right) \\
	& =: & \bm{B}_1 + \bm{B}_2. 
\end{eqnarray*}

The main proof of the theorem has two main steps.  First, we will establish a bound on the moments of the quantity of interest $\|\bm{B}\|_\spnum$.  Next we will use these moments to derive a tail bound on $\|\bm{B}\|_\spnum$, which will lead directly to the RIP statement we seek.  The following two lemmas from the literature will be critical for these two steps.
\begin{lemma}[Lemma 8.2 of~\citep{RAH:2010}]
	\label{lem:RV_1}
	Suppose $\nodenum \ge \spnum$ and suppose we have a sequence of (fixed) vectors $Y_l \in \comps^\memlen$ for $l = 1, \cdots, \nodenum$ such that $\kappa := \max_{l=1, \cdots, M} \|Y_l\|_\infty < \infty$. 
	Let $\{\xi_l\}$ be a Rademacher sequence, i.e., a sequence of i.i.d.\ $\pm 1$ random variables. 
	Then for $p = 1$ and for $p \in \reals$ and $p \ge 2$,
	\begin{eqnarray*}
		\left(\expect{\left\|\sum_{l=1}^\nodenum \xi_l Y_l Y_l^H \right\|_\spnum^p}\right)^{1/p} & \le & C' C^{1/p} \kappa \sqrt{p} \sqrt{\spnum} \log(100\spnum) \sqrt{\log(4\memlen) \log(10\nodenum)} \\
	& &\times	\sqrt{\left\| \sum_{l=1}^\nodenum Y_l Y_l^H \right\|_\spnum},
	\end{eqnarray*}
	where $C, C'$ are universal constants.
\end{lemma}

\begin{lemma}[Adapted from Proposition 6.5 of~\citep{RAH:2010}]
	\label{lem:tailbound}
	Suppose $Z$ is a random variable satisfying
	\begin{eqnarray*}
		\left( \expect{|Z|^p} \right)^{1/p} \le \alpha \beta^{1/p} p^{1/\gamma},
	\end{eqnarray*}
	for all $p \in [p_0,\; p_1]$, and for constants $\alpha, \beta, \gamma, p_0, p_1$. 
	Then, for all $u \in [p_0^{1/\gamma},\; p_1^{1/\gamma}]$,
	\begin{eqnarray*}
		\proba{|Z| \ge e^{1/\gamma} \alpha u} \le \beta e^{-u^{\gamma}/\gamma}. 
	\end{eqnarray*}
\end{lemma}

Armed with this notation and these lemmas, we now prove Theorem~\ref{thm:STMbasic}:
\begin{proof}
	We seek to show that under the conditions on $\nodenum$ in Theorem~\ref{thm:STMbasic}, $\proba{\|\bm{B}\|_\spnum > \ripcd} \le \eta$. 
	Since $\bm{B} = \bm{B}_1 + \bm{B}_2$ and $\{\|\bm{B}_1\|_\spnum \le \ripcd/2\} \cap \{\|\bm{B}_2\|_\spnum \le \ripcd/2\} \subset \{\|\bm{B}\|_\spnum \le \ripcd\}$, then,
	\begin{eqnarray*}
		\proba{\|\bm{B}\|_\spnum > \ripcd} \le \proba{\|\bm{B}_1\|_\spnum > \ripcd/2} + \proba{\|\bm{B}_2\|_\spnum > \ripcd/2}. 
	\end{eqnarray*}
	Thus, it will suffice to bound $\proba{\|\bm{B}_1\|_\spnum > \ripcd/2} \le \eta/2$ since $\bm{B}_2 = \bm{B}_1^{*}$ implies that $\proba{\|\bm{B}_2\|_\spnum > \ripcd/2}\le \eta/2$. In this presentation we let $C, C'$ be some universal constant that may not be the same from line to line.

	To begin, we use Lemma~\ref{lem:RV_1} to bound $E_p := \left(\expect{\|\bm{B}_1\|_\spnum^p}\right)^{1/p}$ by setting $Y_l = \ffvecu_l^{*} X_l$ for $l = 1, \cdots, \frac{\nodenum}{2}$. 
	To meet the conditions of Lemma~\ref{lem:RV_1} we use a standard ``symmetrization'' manipulation (see Lemma 6.7 of~\citep{RAH:2010}).  Specifically, we can write:
\begin{eqnarray*}
	E_p & = & \left(\expect{\| \bm{B}_1 \|_\spnum^p}\right)^{1/p} \\
	& \le & 2 \left(\expect{\left\|  \sum_{l=1}^{\nodenum/2} \xi_l Y_l Y_l^H \right\|_\spnum^p}\right)^{1/p} \\
	& = & 2 \left(\expect{\left\| \sum_{l=1}^{\nodenum/2} \xi_l |\ffvecu_l|^2 X_l X_l^H \right\|_\spnum^p}\right)^{1/p},
\end{eqnarray*}
where now the expectation is over the old random sequence $\{w_l\}$, together with a newly added Rademacher sequence $\{\xi_l\}$. 
Applying the law of iterated expectation and Lemma~\ref{lem:RV_1}, we have for $p \ge 2$:
\begin{eqnarray}
	\label{eq:Ep}
	E_p^p & := & \expect{\| \bm{B}_1 \|_{\spnum}^p} \\
	&\le& 2^p \expect{\expect{\left\| \sum_{l=1}^{\nodenum/2} \xi_l |\ffvecu_l|^2 X_l X_l^H \right\|_\spnum^p \;\vert\; \{w_l\}}} \\
	&\le& \left(2 C' C^{1/p} \sqrt{p} \kappa \sqrt{\spnum} \log (100 \spnum) \sqrt{\log(4 \memlen) \log(5 \nodenum)} \right)^p \expect{\left\| \sum_{l=1}^{\nodenum/2} |\ffvecu_l|^2 X_l X_l^H \right\|_\spnum^{p/2}} \nonumber \\
	&\le& \left(C^{1/p} \sqrt{p} \kappa \sqrt{C' \spnum \log^4(\memlen)}\right)^{p} \expect{\left( \left\| \sum_{l=1}^{\nodenum/2} \left(|\ffvecu_l|^2 X_l X_l^H - \frac{1}{2} \bm{I} \right) \right\|_\spnum + \frac{1}{2} \|I\|_\spnum \right)^{p/2}} \nonumber \\
	&\le& \left(C^{1/p} \sqrt{p} \kappa \sqrt{C' \spnum \log^4(\memlen)}\right)^{p} \sqrt{\expect{\left( \left\| \sum_{l=1}^{\nodenum/2} \left(|\ffvecu_l|^2 X_l X_l^H - \frac{1}{2} \bm{I} \right) \right\|_\spnum + \frac{1}{2} \right)^{p}}} \nonumber \\
	&\le& \left(C^{1/p} \sqrt{p} \kappa \sqrt{C' \spnum \log^4(\memlen)}\right)^{p} \sqrt{\left(E_p + \frac{1}{2} \right)^{p}}. \nonumber
\end{eqnarray}
In the first line above, the inner expectation is over the Rademacher sequence $\{\xi_l\}$ (where we apply Lemma~\ref{lem:RV_1}) while the outer expectation is over the $\{w_l\}$. 
The third line uses the triangle inequality for the $\| \cdot \|_\spnum$ norm, the fourth line uses Jansen's inequality, and the fifth line uses triangle inequality for moments norm (i.e., $(\expect{|X + Y|^p})^{1/p} \le (\expect{|X|^p})^{1/p} + (\expect{|Y|^p})^{1/p}$). 
To get to $\log^4 \memlen$ in the third line, we used our assumption that $\memlen \ge \nodenum$, $\memlen \ge \spnum$ and $\memlen \ge O(1)$ in Theorem~\ref{thm:STMbasic}. 
Now using the definition of $\kappa$ from Lemma~\ref{lem:RV_1}, we can bound this quantity as: 
\begin{eqnarray*}
	\kappa := \max_l \|Y_l\|_\infty = \max_{l} |\ffvecu_l| \|X_l\|_\infty = \frac{1}{\sqrt{\nodenum}} \max_l \|X_l\|_\infty = \frac{1}{\sqrt{\nodenum}} \max_{l,n} |\langle \fftmat_l, \wavmat_n \rangle| \le \frac{\mu(\wavmat)}{\sqrt{M}}. 
\end{eqnarray*} 
Therefore, we have the following implicit bound on the moments of the random variable of interest
\begin{eqnarray*}
	E_p &\le& C^{1/p} \sqrt{p} \sqrt{\frac{C' \spnum \mu(\wavmat)^2 \log^4(\memlen)}{\nodenum}} \sqrt{E_p + \frac{1}{2}}.
\end{eqnarray*}
The above can be written as $E_p \le a_p \sqrt{E_p + \frac{1}{2}}$, where $a_p = C^{1/p} \sqrt{p} \sqrt{\frac{4 C' \spnum \mu(\wavmat)^2 \log^4(\memlen)}{\nodenum}}$. 
By squaring, rearranging the terms and completing the square, we have $E_p \le \frac{a_p^2}{2} + a_p \sqrt{\frac{1}{2} + \frac{a_p^2}{4}}$. 
By assuming $a_p \le \frac{1}{2}$, this bound can be simplified to $E_p \le a_p$. 
Now, this assumption is equivalent to having an upper bound on the range of values of $p$:
\begin{eqnarray*}
	a_p \le \frac{1}{2} &\Leftrightarrow& \sqrt{p} \le \frac{1}{2 C^{1/p}} \sqrt{\frac{\nodenum}{4 C' \spnum \mu(\wavmat)^2 \log^4(\memlen)}} \\
	&\Leftrightarrow& p \le {\frac{\nodenum}{16 C^{2/p} C' \spnum \mu(\wavmat)^2 \log^4(\memlen)}}.
\end{eqnarray*}

Hence, by using Lemma~\ref{lem:tailbound} with $\alpha = \sqrt{\frac{C' \spnum \mu(\wavmat)^2 \log^4(\memlen)}{\nodenum}}$, $\beta = C$, $\gamma = 2$, $p_0 = 2$, and $p_1 = \frac{\nodenum}{16 C^{2/p} C' \spnum \mu(\wavmat)^2 \log^4(\memlen)}$ we obtain the following tail bound for $u \in [\sqrt{2},\; \sqrt{p_1}]$:
\begin{eqnarray*}
	\proba{\| \bm{B}_1 \|_\spnum \ge e^{1/2} \sqrt{\frac{C' \spnum \mu(\wavmat)^2 \log^4(\memlen)}{\nodenum}} u} \le C e^{-u^2/2}. 
\end{eqnarray*}
If we pick $\ripcd < 1$ such that 
\begin{eqnarray}
	e^{1/2} \sqrt{\frac{C' \spnum \mu(\wavmat)^2 \log^4(\memlen)}{\nodenum}} u \le \frac{\ripcd}{2}
	\label{eq:requirementonM}
\end{eqnarray}
and $u$ such that 
\begin{eqnarray*}
	C e^{-u^2/2} \le \frac{\eta}{2} &\Leftrightarrow& u \ge \sqrt{2 \log(2C \eta^{-1})},
\end{eqnarray*}
then we have our required tail bound of $\proba{\|\bm{B}_1\|_\spnum > \ripcd} \le \eta/2$. 
First, observe that Equation~\eqref{eq:requirementonM} is equivalent to having
\begin{eqnarray*}
	\nodenum \ge \frac{C\spnum \mu(\wavmat)^2 \log^4(\memlen) \log(\eta^{-1})}{\ripcd^2}. 
\end{eqnarray*}
Also, because of the limited range of values $u$ can take (i.e., $u \in [\sqrt{2},\; \sqrt{p_1}]$), we require that 
\begin{eqnarray*}
	\sqrt{2 \log(2C \eta^{-1})} & \le & \sqrt{\frac{\nodenum}{16 C^{2/p} C' \spnum \mu(\wavmat)^2 \log^4(\memlen)}} = \sqrt{p_1} \\
	&\Leftrightarrow& M \ge C \spnum \mu(\wavmat)^2 \log^4(\memlen) \log(\eta^{-1}),
\end{eqnarray*}
which, together with the earlier condition on $M$, completes the proof.

\end{proof}



\subsection{RIP with Gaussian feed-forward vectors}
\label{app:RIPz}

In this appendix we extend the RIP analysis of Appendix~\ref{app:BasicRIP} to the case when $\ffvec$ is chosen to be a Gaussian i.i.d.\ vector, as presented in Theorem~\ref{thm:STMwithZ}.
It is unfortunate that with the additional randomness in the feed-forward vector, the same proof procedure as in Theorem~\ref{thm:STMbasic} cannot be used. 
In the proof of Theorem~\ref{thm:STMbasic}, we showed that the random variable $\|Z_1\|_\spnum$ has $p$-th moments that scale like $\alpha \beta^{1/p} p^{1/2}$ (through Lemma~\ref{lem:RV_1}) for a range of $p$ which suggests that it has a sub-gaussian tail (i.e., $\proba{\|Z_1\|_\spnum > u} \le Ce^{-u^2/2}$) for a range of deviations $u$. We then used this tail bound to bound the probability that $\|Z_1\|_\spnum$ exceeds a fixed conditioning $\ripcd$. With Gaussian uncertainties in the feed-forward vector $z$, Lemma~\ref{lem:RV_1} will not yield the required sub-gaussian tail but instead gives us moments estimates that result in sub-optimal scaling of $\nodenum$ with respect to $\memlen$. Therefore, we will instead follow the proof procedure of Theorem 16 from~\citep{Tropp2009a} that will yield the better measurement rate given in Theorem~\ref{thm:STMwithZ}. 

Let us begin by recalling a few notations from the proof of Theorem~\ref{thm:STMbasic} and by introducing further notations that will simplify our exposition later. 
First, recall that we let $X_l^H$ be the $l$-th row of $\fftmat \wavmat$. 
Thus, the $l$-th row of our matrix of interest $\widehat{\ripmat} = \zmat \fftmat \wavmat$ is $\ffvecu_l X_l^H$ where $\ffvecu_l$ is the $l$-th diagonal entry of the diagonal matrix $\zmat$. 
Whereas before, $\ffvecu_l = \frac{1}{\sqrt{\nodenum}}$ for any $l = 1,\cdots, \nodenum$, here it will be a random variable. 
To understand the resulting distribution of $\ffvecu_l$, first note that for the connectivity matrix $\conmat$ to be real, we need to assume that the second $\frac{\nodenum}{2}$ columns of $\conmatev$ are complex conjugates of the first $\frac{\nodenum}{2}$ columns. 
Thus, we can write $\conmatev = \left[\conmatev_R \; | \; \conmatev_R \right] + j \left[ \conmatev_I \;| \; -\conmatev_I \right]$, where $\conmatev_R, \conmatev_I \in \reals^{\nodenum \times \frac{\nodenum}{2}}$. 
Because $\conmatev^H \conmatev = \bm{I}$, we can deduce that $\conmatev_R^T \conmatev_I = \bm{0}$ and that the $\ell_2$ norms of the columns of both $\conmatev_R$ and $\conmatev_I$ are $\frac{1}{\sqrt{2}}$.\footnote{
This can be shown by writing 
\begin{eqnarray*}
\conmatev^H \conmatev 
&=& \left( \left[ \frac{\conmatev_R^T}{\conmatev_R^T} \right] - j \left[ \frac{\conmatev_I^T}{ -\conmatev_I^T} \right] \right) 
\left( \left[\conmatev_R \; | \; \conmatev_R \right] + j \left[ \conmatev_I \;| \; -\conmatev_I \right] \right) \\
&=& \left( \left[ \begin{smallmatrix} \conmatev_R^T \conmatev_R & \conmatev_R^T \conmatev_R \\ \conmatev_R^T \conmatev_R & \conmatev_R^T \conmatev_R \end{smallmatrix} \right] + \left[ \begin{smallmatrix} \conmatev_I^T \conmatev_I & -\conmatev_I^T \conmatev_I \\ -\conmatev_I^T \conmatev_I & \conmatev_I^T \conmatev_I \end{smallmatrix} \right] \right) 
+ j \left( \left[ \begin{smallmatrix} \conmatev_R^T \conmatev_I & -\conmatev_R^T \conmatev_I \\ \conmatev_R^T \conmatev_I & -\conmatev_R^T \conmatev_I \end{smallmatrix} \right] + \left[ \begin{smallmatrix} \conmatev_I^T \conmatev_R & \conmatev_I^T \conmatev_R \\ -\conmatev_I^T \conmatev_R & -\conmatev_I^T \conmatev_R \end{smallmatrix} \right] \right).
\end{eqnarray*}
Then by equating the above to $I + j 0$, we arrive at our conclusion. 
} 

With these matrices $\conmatev_R, \conmatev_I$, let us re-write the random vector  $\ffvecu$ to illustrate its structure. 
Consider the matrix $\widehat{\conmatev} := \left[ \conmatev_R \;|\; \conmatev_I \right] \in \reals^{\nodenum \times \nodenum}$, which is a scaled unitary matrix (because we can check that $\widehat{\conmatev}^T \widehat{\conmatev} = \frac{1}{2} \bm{I}$). 
Next, consider the random vector $\widehat{\ffvec} := \widehat{\conmatev}^T \ffvec$. 
Because $\widehat{\conmatev}$ is (scaled) unitary and $\ffvec$ is composed of i.i.d.\ zero-mean Gaussian random variables of variance $\frac{1}{\nodenum}$, the entries of $\widehat{\ffvec}$ are also i.i.d.\ zero-mean Gaussian random variables, but now with variance $\frac{1}{2\nodenum}$. 
Then, from our definition of $\conmatev$ in terms of $\conmatev_R$ and $\conmatev_I$, for any $l \le \frac{\nodenum}{2}$,  we have $\ffvecu_l = \widehat{\ffvec}_l - j \widehat{\ffvec}_{l + \frac{\nodenum}{2}}$ and for $l > \frac{\nodenum}{2}$, we have $\ffvecu_l = \widehat{\ffvec}_{l - \frac{\nodenum}{2}} + j \widehat{\ffvec}_l$. 
This clearly shows that each of the \emph{first} $\frac{M}{2}$ entries of $\ffvecu$ is made up of 2 i.i.d.\ random variables (one being the real component, the other imaginary), and that the other $\frac{M}{2}$ entries are just complex conjugates of the first $\frac{M}{2}$. 
Because of this, for $l \le \frac{\nodenum}{2}$, $|\ffvecu_l|^2 = |\ffvecu_{l+\frac{\nodenum}{2}}|^2 = \widehat{\ffvec}_l^2 + \widehat{\ffvec}_{l + \frac{\nodenum}{2}}^2$ is the sum of squares of 2 i.i.d.\ Gaussian random variables.  


From the proof of Theorem~\ref{thm:STMbasic}, we also denoted 
\begin{eqnarray*}
	Z &:=& \widehat{\ripmat}^H \widehat{\ripmat} - \bm{I} 
	\;=\; \left(\sum_{l = 1}^{\nodenum/2} |\ffvecu_l|^2 {X_l X_l^H} - \frac{1}{2} \bm{I} \right) + \left(\sum_{l = 1}^{\nodenum/2} |\ffvecu_l|^2 \left(X_l X_l^H\right)^{*} - \frac{1}{2} \bm{I} \right) 
	\;=:\; Z_1 + Z_2.
\end{eqnarray*}
It is again easy to check that $\expect{\sum_{l = 1}^{\nodenum/2} \left(|\ffvecu_l|^2 {X_l X_l^H} \right)} = \expect{\sum_{l = 1}^{\nodenum/2} \left(|\ffvecu_l|^2 \left(X_l X_l^H\right)^{*} \right)} = \frac{1}{2} \bm{I}$. 
Finally, $\widehat{\ripmat}$ has RIP conditioning $\ripcd$ whenever $\|Z\|_\spnum \le \ripcd$ with 
$\displaystyle\|Z\|_\spnum := \sup_{\tiny \mbox{$y$ is \spnum-sparse}} \frac{y^H Z y}{y^H y}$.  

Before moving on to the proof, we first present a lemma regarding the random sequence $|\ffvec_l|^2$ that will be useful in the sequel.
\begin{lemma}
	\label{lem:max_ffvecu} 
	Suppose for $l = 1, \cdots, \frac{\nodenum}{2}$, $|\ffvecu_l|^2 = \widehat{\ffvec}_l^2 + \widehat{\ffvec}_{l + \nodenum/2}^2$ where $\widehat{\ffvec}_l$ for $l = 1, \cdots, \nodenum$ is a sequence of i.i.d.\ zero-mean Gaussian random variables of variance $\frac{1}{2\nodenum}$. Also suppose that $\eta \leq 1$ is a fixed probability. 
	For the random variable $\max_{l = 1, \cdots, \nodenum/2} |\ffvecu_l|^2$, we have the following bounds on the expected value and tail probability of this extreme value:
	\begin{eqnarray}
		\expect{\max_{l = 1, \cdots, \nodenum/2} |\ffvecu_l|^2} &\le& \frac{1}{\nodenum} \left(\log\left( \frac{C_1 \nodenum}{2} \right) + 1\right), \label{eq:E_max_z}\\
		\proba{\max_{l = 1, \cdots, \nodenum/2} |\ffvecu_l|^2 > \frac{C_2 \log\left(C_2' \nodenum \eta^{-1}\right)}{\nodenum}} &\le& \eta. \label{eq:P_max_z}
	\end{eqnarray}
\end{lemma}

\begin{proof}
	To ease notation, every index $l$ used as a variable for a maximization will be taken over the set $l=1,\dots,\frac{M}{2}$ without explicitly writing the index set.
	To calculate $\expect{\max_l |\ffvecu_l|^2}$, we use the following result that allows us to bound the expected value of a positive random variable by its tail probability (see Proposition 6.1 of [Rauhut]):
	\begin{eqnarray}
		\expect{\max_l |\ffvecu_l|^2} = \int_{0}^\infty \proba{\max_l |\ffvecu_l|^2 > u} du.
		\label{eq:E_max_int}
	\end{eqnarray}
	Using the union bound, we have the estimate $\proba{\max_l |\ffvecu_l|^2 > u} \le \frac{\nodenum}{2} \proba{|\ffvecu_1|^2 > u}$ (since the $|\ffvecu_l|^2$ are identically distributed). 
	Now, because $|\ffvecu_1|^2$ is a sum of squares of two Gaussian random variables and thus is a (generalized) $\chi^2$ random variable with 2 degrees of freedom (which we shall denote by $\chi_2$),\footnote{
	The pdf of a $\chi^2$ random variable $\chi_q$ with $q$ degrees of freedom is given by $p(x) = \frac{1}{2^{q/2}\Gamma(q/2)} x^{q/2 - 1} e^{-x/2}$. Therefore, it's tail probability can be obtained by integration: $\proba{\chi_q > u} = \int_u^{\infty} p(x) dx$. 
	} we have
	\begin{eqnarray*} 
		\proba{|\ffvecu_1|^2 > u} 
		= \proba{\chi_2 > 2 \nodenum u} = \frac{1}{\Gamma(1)} e^{\frac{-2 \nodenum u}{2}} = C_1 e^{-\nodenum u}, \label{eq:chi_square_cum}
	\end{eqnarray*} 
	where 
	$\Gamma(\cdot)$ is the Gamma function and the $2 \nodenum u$ appears instead of $u$ in the exponential because of the standardization of the Gaussian random variables (initially of variance $\frac{1}{2 \nodenum}$). 
	To proceed, we break the integral in \eqref{eq:E_max_int} into 2 parts. 
	To do so, notice that if $u < \frac{1}{\nodenum} \log\left( \frac{C_1 \nodenum}{2} \right)$, then the trivial upper bound of $\proba{\max_l |\ffvecu_l|^2 > u} \le 1$ is a better estimate than $\frac{C_1 \nodenum}{2} e^{-\nodenum u}$. 
	In other words, our estimate for the tail bound of $\max_l |\ffvecu_l|^2$ is not very good for small $u$ but gets better with increasing $u$. 
	Therefore, we have
	\begin{eqnarray*}
		\expect{\max_l |\ffvecu_l|^2} &\le& \int_{0}^{\frac{1}{\nodenum} \log\left( \frac{C_1 \nodenum}{2} \right)} 1 \; du + \int_{\frac{1}{\nodenum} \log\left( \frac{C_1 \nodenum}{2} \right)}^\infty \frac{C_1 \nodenum}{2} e^{-\nodenum u} \; du \\
		&=& \frac{1}{\nodenum} \log\left( \frac{C_1 \nodenum}{2} \right) - \frac{C_1 \nodenum}{2} \left.\left[\frac{1}{\nodenum} e^{- \nodenum u} \right]\right|_{\frac{1}{\nodenum} \log\left( \frac{C_1 \nodenum}{2} \right)}^\infty \\
		&=& \frac{1}{\nodenum} \log\left( \frac{C_1 \nodenum}{2} \right) + \frac{C_1}{2} e^{- \log\left( \frac{C_1 \nodenum}{2} \right)} = \frac{1}{\nodenum} \left(\log\left( \frac{C_1 \nodenum}{2} \right) + 1\right). 
	\end{eqnarray*}
	This is the bound in expectation that we seek for in Equation~\eqref{eq:E_max_int}. 
		
	In the second part of the proof that follows, $C, C'$ denote universal constants. 
	Essentially, we will want to apply Lemma~\ref{lem:tailbound} that is used in Appendix~\ref{app:BasicRIP} to obtain our tail bound. 
	In the lemma, the tail bound of a random variable $X$ can be estimated once we know the moments of $X$. 
	Therefore, we require the moments of the random variable $\max_l |\ffvecu_l|^2$. 
	For this, for any $p > 0$, we use the following simple estimate:
	\begin{eqnarray}
		\expect{\max_l |\ffvecu_l|^{2p}} \le \frac{\nodenum}{2} \max_l \expect{|\ffvecu_l|^{2p}} = \frac{\nodenum}{2} \expect{|\ffvecu_1|^{2p}},
		\label{eq:crude_max_bound}
	\end{eqnarray} 
	where the first step comes from writing the expectation as an integral of the cumulative distribution (as seen in Equation~\eqref{eq:E_max_int}) and taking the union bound, and the second step comes from the fact that the $|\ffvecu_l|^2$ are identically distributed. 
	Now, $|\ffvecu_1|^{2}$ is a sub-exponential random variable since it is a sum of squares of Gaussian random variables~\citep{Vershynin2011}.\footnote{
	A sub-exponential random variable is a random variable whose tail probability is bounded by $\exp^{-Cu}$ for some constant $C$. Thus, a $\chi^2$ random variable is a specific instance of a sub-exponential random variable. 
	} 
	Therefore, for any $p > 0$, it's $p$-th moment can be bounded by
	\begin{eqnarray*}
		\expect{|\ffvecu_1|^{2p}}^{1/p} \le \frac{C'}{\nodenum} C^{1/p} p,
	\end{eqnarray*}
	where the division by $\nodenum$ comes again from the variance of the Gaussian random variables that make up $|\ffvecu_1|^2$. 
	Putting this bound with Equation~\eqref{eq:crude_max_bound}, we have the following estimate for the $p$-th moments of $\max_l |\ffvecu_l|^2$:\footnote{We remark that this bound gives a worse estimate for the expected value as that calculated before because of the crude bound given by Equation~\eqref{eq:crude_max_bound}.}
	\begin{eqnarray*}
		\expect{\max_l |\ffvecu_l|^{2p}}^{1/p} \le \frac{C'}{\nodenum} \left(\frac{CM}{2} \right)^{1/p} p.
	\end{eqnarray*}
	Therefore, by Lemma~\ref{lem:tailbound} with $\alpha = \frac{C'}{\nodenum}$, $\beta = \frac{CM}{2}$, and $\gamma = 1$, we have
	\begin{eqnarray*}
		\proba{\max_l |\ffvecu_l|^{2} > \frac{eC'u}{\nodenum}} \le \frac{C\nodenum}{2}e^{-u}. 
	\end{eqnarray*}
	By choosing $u = \log\left(\frac{C\nodenum}{2}\eta^{-1} \right)$, we have our desired tail bound of 
	\begin{gather}
		\proba{\max_l |\ffvecu_l|^2 > \frac{C_2 \log\left( C_2' \nodenum \eta^{-1} \right)}{\nodenum}} \le \eta. \nonumber
	\end{gather}
\end{proof}

Armed with this lemma, we can now turn out attention to the main proof.  
As stated earlier, this follows essentially the same form as~\citep{Tropp2009a} with the primary difference of including the results from Lemma~\ref{lem:max_ffvecu}.  
As before, because $\proba{\|Z\|_\spnum > \ripcd} \le \proba{\|Z_1\|_\spnum > \ripcd/2} + \proba{\|Z_2\|_\spnum > \ripcd/2}$ with $Z_2 = Z_1^{*}$, we just have to consider bounding the tail bound $\proba{\|Z_1\|_\spnum > \ripcd/2}$. 
This proof differs from that in Appendix~\ref{app:BasicRIP} in that here, we will first show that $\expect{\|Z_1\|_\spnum}$ is small when $\nodenum$ is large enough and then show that $Z_1$ does not differ much from $\expect{\|Z_1\|_\spnum}$ with high probability. 

\subsubsection*{Expectation}

In this section, we will show that $\expect{\|Z_1\|_\spnum}$ is small. 
This will basically follow from Lemma~\ref{lem:RV_1} in Appendix~\ref{app:BasicRIP} and Equation~\eqref{eq:E_max_z} in Lemma~\ref{lem:max_ffvecu}. 
To be precise, the remainder of this section is to prove:
\begin{thm}
	\label{thm:expect}
	Choose any $\ripcd' \le \frac{1}{2}$. If $\nodenum \ge \frac{C_3 \spnum \mu(\wavmat)^2 \log^5 \memlen}{\ripcd'^2}$, then $\expect{\|Z\|_\spnum} \le \ripcd'$. 
\end{thm}

\begin{proof}
	Again, $C$ is some universal constant that may not be the same from line to line. 
	We follow the same symmetrization step found in the proof in Appendix~\ref{app:BasicRIP} to arrive at:
	\begin{eqnarray*}
		E:= \expect{\| Z_1 \|_\spnum} &\le& 2\expect{\expect{\left\| \sum_{l=1}^{\nodenum/2} \xi_l |\ffvecu_l|^2 X_l X_l^H \right\|_\spnum \vert \{w_l\}, \ffvecu}},
	\end{eqnarray*}
	where the outer expectation is over the Rademacher sequence $\{\xi_l\}$ and the inner expectation is over the random ``frequencies'' $\{w_l\}$ and feed-forward vector $\ffvecu$. 
	As before, for $l = 1, \cdots, \frac{M}{2}$, we set $Y_l = \ffvecu_l^{*} X_l$. 
	Observe that by definition $\kappa := \max_{l=1, \cdots, \nodenum/2} \|Y_l\|_\infty =  \max_{l} |\ffvecu_l| \|X_l\|_\infty$ and thus is a random variable. 
	We then use Lemma~\ref{lem:RV_1} with $p = 1$ to get
\begin{eqnarray}
		E &\le& 2 C \sqrt{\spnum} \log(100\spnum) \sqrt{\log(4\memlen) \log(5\nodenum)} \expect{\kappa \sqrt{\left\| \sum_{l=1}^{\nodenum/2} |\ffvecu_l|^2 X_l X_l^H \right\|_\spnum}} \nonumber \\
	&\le& \sqrt{4 C \spnum \log^4(\memlen)} \sqrt{\expect{\kappa^2}} \sqrt{\expect{\left\| \sum_{l=1}^{\nodenum/2} |\ffvecu_l|^2 X_l X_l^H \right\|_\spnum}} \nonumber \\
	&\le& \sqrt{4 C \spnum \log^4(\memlen)} \sqrt{\expect{\kappa^2}} \sqrt{E + \frac{1}{2}}, \label{eqn:branchpoint}
\end{eqnarray}
where the second line uses the Cauchy-Schwarz inequality for expectations and the third line uses triangle inequality. 
Again, to get to $\log^4 \memlen$ in the second line, we used our assumption that $\memlen \ge \nodenum$, $\memlen \ge \spnum$ and $\memlen \ge O(1)$ in Theorem~\ref{thm:STMwithZ}. 
It therefore remains to calculate $\expect{\kappa^2}$. 
Now, $\kappa = \max_{l} |\ffvecu_l| \|X_l\|_\infty \le \max_l |\ffvecu_l| \max_l \|X_l\|_\infty$. 
First, we have $\max_l \|X_l\|_\infty = \max_{l,n} |\langle \fftmat_l, \wavmat_n \rangle| \le \mu(\wavmat)$. 
Next, \eqref{eq:E_max_z} in Lemma~\ref{lem:max_ffvecu} tells us that $\expect{\max_{l = 1, \cdots, \nodenum/2} |\ffvecu_l|^2} \le \frac{1}{\nodenum} \left(\log\left( \frac{C_1 \nodenum}{2} \right) + 1 \right)$. 
Thus, we have $\expect{\kappa^2}  \le \frac{\mu(\wavmat)^2}{\nodenum} \left(\log\left( \frac{C_1 \nodenum}{2} \right) + 1\right)$. 
Putting everything together, we have
\begin{eqnarray*}
	E= \expect{\| Z_1 \|_\spnum} \le \sqrt{\frac{C \spnum \log^4(\memlen) \left(\log\left( \frac{C_1 \nodenum}{2} \right) + 1 \right) \mu(\wavmat)^2}{\nodenum}} \sqrt{E + \frac{1}{2}}. 
\end{eqnarray*}
Now, the above can be written as $E \le a \sqrt{E + \frac{1}{2}}$, where $a = \sqrt{\frac{C \spnum \log^4(\memlen) \left(\log\left( \frac{C_1 \nodenum}{2} \right) + 1 \right) \mu(\wavmat)^2}{\nodenum}}$. 
By squaring it, rearranging the terms and completing the squares, we have $E \le \frac{a^2}{2} + a \sqrt{\frac{1}{2} + \frac{a^2}{4}}$. 
By supposing $a \le \frac{1}{2}$, this can be simplified as $E \le a$. 
To conclude, let us choose $\nodenum$ such that $a \le \ripcd'$ where $\ripcd' \le \frac{1}{2}$ is our pre-determined conditioning (which incidentally fulfills our previous assumption that $a \le \frac{1}{2}$). 
By applying the formula for $a$, we have that if $\nodenum \ge \frac{ C_3 \spnum \mu(\wavmat)^2 \log^5(\memlen)}{\ripcd'^2} $, then $E \le \ripcd'$. 
\end{proof}

\subsubsection*{Tail Probability}

To give a probability tail bound estimate to $Z_1$, we use the following lemma found in~\citep{Tropp2009a,RAH:2010}:
\begin{lemma}
	\label{lem:tail_prob}
	Suppose $Y_l$ for $l = 1, \cdots, \nodenum$ are independent, symmetric random variables such that $\|Y_l\|_\spnum \le \zeta < \infty$ almost surely. 
	Let $Y = \sum_{l= 1}^\nodenum Y_l$. 
	Then for any $u,t > 1$, we have
	\begin{eqnarray*}
		\proba{\|Y\|_\spnum > C (u \expect{\|Y\|_\spnum} + t \zeta)} \le e^{-u^2} + e^{-t}.
	\end{eqnarray*}
\end{lemma}

The goal of this section is to prove:
\begin{thm}
	Pick any $\ripcd \le  \frac{1}{2}$ and suppose $\memlen^{-\log^4\left(\memlen\right)} \le \eta \le \frac{1}{e}$. Suppose $\nodenum \ge \frac{C_4 \spnum \mu(\wavmat)^2 \log^5 \memlen \log \eta^{-1}}{\ripcd^2}$, then $\proba{\|Z_1\|_\spnum > \ripcd} \le 8 \eta$. 
\end{thm}

\begin{proof}
To use Lemma~\ref{lem:tail_prob}, we want $Y_l$ to look like the summands of 
\begin{gather}
	Z_1 = \sum_{l = 1}^{\nodenum/2} \left(|\ffvecu_l|^2 {X_l X_l^H} - \expect{|\ffvecu_l|^2 {X_l X_l^H}} \right). \nonumber
\end{gather}
However, this poses several problems. First, they are \emph{not} symmetric\footnote{A random variable $X$ is symmetric if $X$ and $-X$ has the same distribution.} and thus, we need to symmetrize it by defining 
\begin{eqnarray*}
	\widetilde{Y}_l &=& |\ffvecu_l|^2 {X_l X_l^H} - |\ffvecu_l'|^2 X_l' (X_l')^H  \\
	&\sim& \xi_l \left( |\ffvecu_l|^2 {X_l X_l^H} - |\ffvecu_l'|^2 X_l' (X_l')^H \right)
\end{eqnarray*} 
where $\ffvecu', X_l'$ are independent copies of $\ffvecu$ and $X_l$ respectively, and $\xi_l$ is an independent Rademacher sequence. 
Here, the relation $X \sim Y$ for two random variables $X,Y$ means that $X$ has the same distribution as $Y$. 
To form $\widetilde{Y}_l$, what we have done is take each summand of $Z_1$ and take it's difference with an independent copy of itself. 
Because $\widetilde{Y}_l$ is symmetric, adding a Rademacher sequence does not change its distribution and this sequence is only introduced to resolve a technicality that will arise later on. 
If we let $\widetilde{Y} := \sum_{l=1}^{M/2} \widetilde{Y}_l$, then the random variables $\widetilde{Y}$ (symmetrized) and $Z_1$ (un-symmetrized) are related via the following estimates~\citep{RAH:2010}:
\begin{eqnarray}
	\expect{\|\widetilde{Y}\|_\spnum} &\le& 2\expect{\|Z_1\|_\spnum}, \label{eq:EwtY_EZ} \\
	\proba{\|Z_1\|_\spnum > 2 \expect{\|Z_1\|_\spnum} + u} &\le& 2\proba{\|\widetilde{Y}\|_\spnum > u}. \label{eq:PZ_PwtY}
\end{eqnarray}

However, a second condition imposed on $Y_l$ in Lemma~\ref{lem:tail_prob} is that $\|Y_l\|_\spnum \le \zeta < \infty$ almost surely. 
Because of the unbounded nature of the Gaussian random variables $\ffvecu_l$ and $\ffvecu_l'$ in $\widetilde{Y}_l$, this condition is not met. 
Therefore, we need to define a $Y_l$ that is conditioned on the event that these Gaussian random variables are bounded. 
To do so, define the following event:
\begin{eqnarray*}
	F = \left\{\max\left\{\max_{l} |\ffvecu_l|^2,\; \max_{l} |\ffvecu_l'|^2 \right\} \le \frac{C_2 \log \left( C_2' \nodenum \eta^{-1} \right)}{\nodenum} \right\}.
\end{eqnarray*}

Using Equation~\eqref{eq:P_max_z} in Lemma~\ref{lem:max_ffvecu}, we can calculate $\proba{F^{c}}$, where $F^{c}$ is the complementary event of $F$:
\begin{eqnarray*}
	\proba{F^{c}} &=& \proba{\max\left\{\max_{l} |\ffvecu_l|^2,\; \max_{l} |\ffvecu_l'|^2 \right\} > \frac{C_2 \log \left( C_2' \nodenum \eta^{-1} \right)}{\nodenum}} \\
	&\le& \proba{\max_{l} |\ffvecu_l|^2 > \frac{C_2 \log \left( C_2' \nodenum \eta^{-1} \right)}{\nodenum}} + \proba{\max_{l} |\ffvecu_l'|^2 > \frac{C_2 \log \left( C_2' \nodenum \eta^{-1} \right)}{\nodenum}} \\
	& \le & 2 \eta.
\end{eqnarray*}

Conditioned on event $F$, the $\|\cdot\|_\spnum$ norm of $\widetilde{Y}_l$ is well-bounded:
\begin{eqnarray*}
	\left\| \widetilde{Y}_l \right\|_\spnum 
	&=& \left\| |\ffvecu_l|^2 {X_l X_l^H} - |\ffvecu_l'|^2 X_l' (X_l')^H \right\|_\spnum
	\;\le\; 2 \max\left\{\max_{l} |\ffvecu_l|^2,\; \max_{l} |\ffvecu_l'|^2 \right\} \left\| X_l X_l^H \right\|_\spnum \\
	&=& \frac{2 C_2 \log \left( C_2' \nodenum \eta^{-1} \right)}{\nodenum} \sup_{\tiny \mbox{$y$ is $\spnum$-sparse}} \left\{\frac{y^H X_l X_l^H y}{y^H y} \right\} \\
	 & \le & \frac{2 C_2 \log \left( C_2' \nodenum \eta^{-1} \right)}{\nodenum} \sup_{\tiny \mbox{$y$ is $\spnum$-sparse}} \left\{\|X_l\|_\infty^2 \frac{\|y\|_1^2}{\|y\|_2^2}\right\} \\
	&\le& \frac{2 \spnum C_2 \log \left( C_2' \nodenum \eta^{-1} \right)}{\nodenum} \max_l \|X_l\|_\infty^2
	\;\le\; \frac{C\spnum \mu(\wavmat)^2 \log \left( C_2' \nodenum \eta^{-1} \right)}{\nodenum} := \zeta, 
\end{eqnarray*}
where in the last line we used the fact that the ratio between the $\ell_1$ and $\ell_2$ norms of an $\spnum$-sparse vector is $\spnum$, and the estimate we derived for $\max_l \|X_l\|_\infty^2$ in Appendix~\ref{app:BasicRIP}.

We now define a new random variable that is a truncated version of $\widetilde{Y}_l$ which takes for value 0 whenever we fall under event $F^{c}$, i.e.,
\begin{eqnarray*}
	Y_l := \widetilde{Y}_l \; \mathbb{I}_{F} = \xi_l \left( ||\ffvecu_l|^2 {X_l X_l^H} - |\ffvecu_l'|^2 X_l' (X_l')^H \right) \mathbb{I}_{F},
\end{eqnarray*}
where $\mathbb{I}_{F}$ is the indicator function of event $F_l$. 
If we define $Y = \sum_{l=1}^{M/2} Y_l$, then the random variables $Y$ (truncated) and $\widetilde{Y}$ (un-truncated) are related by~\citep{Tropp2009a} (see also Lemma 1.4.3 of~\citep{DeLaPena1999})
\begin{eqnarray}
	\proba{\|\widetilde{Y}\|_\spnum > u} \le \proba{\|Y\|_\spnum > u} + \proba{F^{c}}. 
	\label{eq:PwtY_PY}
\end{eqnarray}
When $\ffvecu, \ffvecu', X_l, X_l'$ are held constant so only the Rademacher sequence $\xi_l$ is random, then the contraction principle~\citep{Tropp2009a,ledoux1991probability} tells us that $\expect{\|Y\|_\spnum} \le \expect{\|\widetilde{Y}\|_\spnum}$.  Note that the sole reason for introducing the Rademacher sequences is for this use of the contraction principle. 
As this holds point-wise for all $\ffvecu, \ffvecu', X_l, X_l'$, we have 
\begin{eqnarray}
	\expect{\|Y\|_\spnum} \le \expect{\|\widetilde{Y}\|_\spnum}. 
	\label{eq:EY_EwtY}
\end{eqnarray} 

We now have all the necessary ingredients to apply Lemma~\ref{lem:tail_prob}. 
First, by choosing $\ripcd' \le \frac{1}{2}$, from Theorem \ref{thm:expect}, we have that $\expect{\|Z\|_\spnum} \le \ripcd'$ whenever $\nodenum \ge \frac{C_3 \spnum \mu(\wavmat)^2 \log^5 \memlen}{\ripcd'^2}$. 
Thus, by chaining \eqref{eq:EY_EwtY} and \eqref{eq:EwtY_EZ}, we have
\begin{eqnarray*}
	\expect{\|Y\|_\spnum} \le \expect{\|\widetilde{Y}\|_\spnum} \le 2 \expect{\|Z_1\|_\spnum} \le 2 \ripcd'.
\end{eqnarray*}
Also, with this choice of $\nodenum$, we have 
\begin{eqnarray*}
	\zeta = \frac{C \spnum \mu(\wavmat)^2 \log \left( C_2' \nodenum \eta^{-1} \right)}{\nodenum} \le \frac{C \ripcd'^2 \log \left( C_2' \nodenum \eta^{-1} \right)}{\log^5 \memlen}.
\end{eqnarray*} 
Using these estimates for $\zeta$ and $\expect{\|Y\|_\spnum}$, and choosing $u = \sqrt{\log \eta^{-1}}$ and $t = \log \eta^{-1}$, Lemma~\ref{lem:tail_prob} says that
\begin{eqnarray*}
	\proba{\|Y\|_\spnum > C' \left( 2 \ripcd' \sqrt{\log \eta^{-1}} + \frac{C \ripcd'^2 \log \left( C_2' \nodenum \eta^{-1} \right) \log \eta^{-1}}{\log^5 \memlen} \right)} \le 2 \eta. 
\end{eqnarray*}
Then, using the relation between the tail probabilities of $Y$ and $\widetilde{Y}$ \eqref{eq:PwtY_PY} together with our estimate for $\proba{F^{c}}$, we have
\begin{eqnarray*}
	 \proba{\|\widetilde{Y}\|_\spnum > C' \left( 2 \ripcd' \sqrt{\log \eta^{-1}} + \frac{C \ripcd'^2 \log \left( C_2' \nodenum \eta^{-1} \right) \log \eta^{-1}}{\log^5 \memlen} \right)} \le 2 \eta + \proba{F^{c}} \le 4 \eta.
\end{eqnarray*}
Finally, using the relation between the tail probabilities of $\widetilde{Y}$ and $Z$ \eqref{eq:PZ_PwtY}, we have
\begin{eqnarray*}
	\proba{\|Z_1\|_\spnum > 2 \ripcd' + 2 C' \ripcd' \sqrt{\log \eta^{-1}} + \frac{C C' \ripcd'^2 \log \left( C_2' \nodenum \eta^{-1} \right) \log \eta^{-1}}{\log^5 \memlen}}
	\le 8 \eta,
\end{eqnarray*}
where we used the fact that $\expect{\|Z_1\|_\spnum} \le \ripcd'$. 
Then, for a pre-determined conditioning $\ripcd \le \frac{1}{2}$, pick $\ripcd' =  \frac{\ripcd}{3C''\sqrt{\log \eta^{-1}}}$ for a constant $C''$ which will be chosen appropriately later. 
With this choice of $\ripcd'$ and with our assumptions that $\ripcd \le \frac{1}{2}$ and $\eta \le \frac{1}{e}$, the three terms in the tail bound becomes
\begin{eqnarray*}
	2 \ripcd' &=& \frac{\ripcd}{3C''\sqrt{\log \eta^{-1}}} \le \frac{1}{C''}\frac{\ripcd}{3}, \\
	2 C' \ripcd' \sqrt{\log \eta^{-1}} &=& \frac{2C'}{C''}\frac{\ripcd}{3}, \\
	\frac{C C' \ripcd'^2 \log \left( C_2' \nodenum \eta^{-1} \right) \log \eta^{-1}}{\log^5 \memlen} &=& \frac{C C' \ripcd^2 (\log (C_2' \nodenum) + \log \eta^{-1})}{9\left(C''\right)^2\log^5 \memlen} \\
	& \le & \frac{C C' (\log (C_2' \nodenum) + \log \eta^{-1})}{3\left(C''\right)^2\log^5 \memlen} \frac{\ripcd}{3}.
\end{eqnarray*}
As for the last term, if $\eta \ge \frac{1}{C_2' \nodenum}$, then $\frac{C C' (\log (C_2' \nodenum) + \log \eta^{-1})}{3\left(C''\right)^2\log^5 \memlen} \le \frac{2 C C'\log (C_2' \nodenum)}{3\left(C''\right)^2\log^5 \memlen} \le \frac{2C C'}{3\left(C''\right)^2}$ (where we further supposed that $N \ge O(1)$). 
If $\memlen^{-\log^4 \memlen} \le \eta \le \frac{1}{C_2' \nodenum}$ (where the lower bound is from the theorem assumptions), then  $\frac{C C' (\log (C_2' \nodenum) + \log \eta^{-1})}{3\left(C''\right)^2\log^5 \memlen} \le \frac{2 C C'\log \eta^{-1}}{3\left(C''\right)^2\log^5 \memlen} \le \frac{2C C'}{3\left(C''\right)^2}$. 
By choosing $C''$ appropriately large, we then have
\begin{eqnarray*}
	\proba{\|Z_1\|_\spnum > \frac{\ripcd}{3} + \frac{\ripcd}{3} + \frac{\ripcd}{3}} \le 8 \eta.
\end{eqnarray*}
Putting the formula for $\ripcd'$ into $\nodenum \ge \frac{C_3 \spnum \mu(\wavmat)^2 \log^5 \memlen}{\ripcd'^2}$ completes the proof. 
\end{proof}

\subsection{Derivation of recovery bound for infinite length inputs}
\label{app:HistErr}

In this appendix we derive the bound in Equation~\eqref{eqn:decay_bound} of the main text. The approach we take is to bound the individual components of Equation~\eqref{eqn:DecRec} of the main text.  As the noise term due to noise in the inputs is unaffected, we will bound the noise term due to the unrecovered signal (the first term in Equation~\eqref{eqn:DecRec} of the main text) by the component of the input history that is beyond the attempted recovery, and we will bound the signal approximation term (the second term in Equation~\eqref{eqn:DecRec} of the main text) by the quality of the signal recovery possible in the attempted recovery length.  In this way we can observe how different properties of the system and input sequence affect signal recovery. 

To bound the first term in Equation~\eqref{eqn:DecRec} of the main text (i.e., the omission errors due to inputs beyond the recovery window), we first write the current state at any time \memlenf as
\begin{gather}
	\nodevect{\memlenf} = \sum_{\tvar = 0}^{\memlenf} {\conmat^{\memlenf-\tvar}\ffvec\sigvart{\tvar}}. \nonumber
\end{gather}
We only wish to recover the past $\memlen\leq \memlenf$ time steps, so we break up the summation into components of the current state due to ``signal'' (i.e., signal we attempt to recover) and ``noise'' (i..e, older signal we omit from the recovery):
\begin{gather}
	\nodevect{\memlenf} = \sum_{\tvar=\memlenf-\memlen+1}^{\memlenf}{\conmat^{\memlenf-\tvar}\ffvec\sigvart{\tvar}}  + \sum_{n = 0}^{\memlenf-\memlen} {\conmat^{\memlenf-\tvar}\ffvec\sigvart{\tvar}} \nonumber	\\
	= \sum_{\tvar=\memlenf-\memlen+1}^{\memlenf}{\conmat^{\memlenf-\tvar}\ffvec\sigvart{\tvar}} + {\bm{\epsilon}} \nonumber \\
	= \ripmat\sigvec + {\bm{\epsilon}}_2. \nonumber
\end{gather}

From here we can see that the first summation is the matrix multiply $\ripmat\sigvec$ as is discussed in the paper. The second summation here, ${\bm{\epsilon}}_2$, essentially acts as an additional noise term in the recovery. We can further analyze the effect of this noise term by understanding that ${\bm{\epsilon}}_2$ is bounded for well behaved input sequences \sigvart{\tvar} (in fact all that is needed is that the maximum value or the expected value and variance are reasonably bounded) when the eigenvalues of \conmat are of magnitude $\decvar \leq 1$. We can explicitly calculate the worst case scenario bounds on the norm of $\bm{\epsilon}_2$,
\begin{eqnarray}
	\norm{\sum_{\tvar = 0}^{\memlenf-\memlen} {\conmat^{\memlenf-\tvar}\ffvec\sigvart{\tvar}}}_2 & \leq & \norm{\sum_{\tvar = 0}^{\memlenf-\memlen} {\conmatev\left(\decvar\conmated\right)^{\memlenf-\tvar}\conmatev^{-1}\ffvec\sigvart{\tvar}}}_2   \nonumber \\
		& \leq & \norm{\conmatev}_2\norm{\sum_{\tvar = 0}^{\memlenf-\memlen} {\left(\decvar\conmated\right)^{\memlenf-\tvar}\conmatev^{-1}\ffvec\sigvart{\tvar}}}_2, \nonumber
\end{eqnarray}
where $\conmated = \mbox{diag}(d_1,\hdots,d_{\nodenum})$ is the diagonal matrix containing the normalized eigenvalues of \conmat. If we assume that \ffvec is chosen as mentioned as in Section~\ref{sec:FiniteSTM} so that $\conmatev^{-1}\ffvec = \left(1/\sqrt{\nodenum}\right)\bm{1}$, the eigenvalues of \conmat are uniformly spread around a complex circle of radius \decvar, and that \sigvart{\tvar} $\leq$ \sigmax for all \tvar, then we can bound this quantity as
\begin{eqnarray}
	\norm{\sum_{\tvar = 0}^{\memlenf-\memlen} {\conmat^{\memlenf-\tvar}\ffvec\sigvart{\tvar}}}_2 & \leq & \frac{\sigmax}{\sqrt{\nodenum}}\norm{\conmatev}_2\norm{\sum_{\tvar = 0}^{\memlenf-\memlen} {\left(\decvar\conmated\right)^{\memlenf-\tvar}\bm{1}}}_2 \nonumber \\
	& = & \frac{\sigmax}{\sqrt{\nodenum}}\norm{\conmatev}_2\norm{ \left[ \begin{matrix} \sum_{\tvar = 0}^{\memlenf-\memlen} \decvar^{\memlenf-\tvar}d_{1}^{\memlenf-\tvar} \\ \vdots \\ \sum_{\tvar = 0}^{\memlenf-\memlen} \decvar^{\memlenf-\tvar}d_{\nodenum}^{\memlenf-\tvar} \end{matrix} \right] }_2 \nonumber \\
	& = & \frac{\sigmax}{\sqrt{\nodenum}}\norm{\conmatev}_2 \sqrt{ \sum_{k = 1}^{\nodenum}\left| \sum_{\tvar = 0}^{\memlenf-\memlen}\decvar^{\memlenf-\tvar} d_{k}^{\memlenf-\tvar} \right|^2 } \nonumber \\
	& \leq & \sigmax\norm{\conmatev}_2 \left| \sum_{\tvar = 0}^{\memlenf-\memlen}\decvar^{\memlenf-\tvar} \right| \nonumber \\
	& \leq & \sigmax\norm{\conmatev}_2 \left| \frac{\decvar^{\memlen} - \decvar^{\memlenf}}{1 - \decvar} \right| \nonumber
\end{eqnarray}
where $d_{k}$ is the $k^{th}$ normalized eigenvalue of \conmat. 
In the limit of large input signal lengths ($\memlenf \rightarrow \infty$), we have $\memlenf \gg \memlen$ and so $q^\memlen \gg q^{\memlenf}$, which leaves the approximate expression 
\begin{gather}
	\norm{{\bm{\epsilon}}_2}_2 \leq \sigmax\norm{\conmatev}_2 \left| \frac{\decvar^{\memlen}}{1 - \decvar} \right|. \nonumber
\end{gather}

To bound the second term in Equation~\eqref{eqn:DecRec} of the main text (i.e., the signal approximation errors due to imperfect recovery), we must characterize the possible error between the signal (which is \spnum-sparse) and the approximation to the signal with the \spnumf largest coefficients.  In the worst case scenario, there are $\spnum-\spnumf+1$ coefficients that cannot be guaranteed to be recovered by the RIP conditions, and these coefficients all take the maximum value \sigmax.  In this case, we can bound the signal approximation error as stated in the main text:
\begin{eqnarray}
	\frac{\beta}{\sqrt{\spnumf}}\left\|\bm{s} - \bm{s}_{S^{\ast}} \right\|_1 & \leq & \frac{\beta}{\sqrt{\spnumf}}\sum_{\tvar = \spnumf+1}^{\spnum}\left|\decvar^{\tvar}\sigmax\right| \nonumber \\
	& = & \frac{\beta\sigmax}{\sqrt{\spnumf}}\left(\frac{\decvar^{\spnumf} - \decvar^{\spnum}}{1 - \decvar}\right).	 \nonumber
\end{eqnarray}

In the case where noise is present, we can also bound the total power of the noise term, 
\begin{gather}
	\alpha\left\|\sum_{k = 0}^{\memlen+\memlen^{\ast}} \conmat^{k}\ffvec\widetilde{\epsilon}[k]\right\|_2^2, \nonumber
\end{gather}
using similar steps. Taking $\epsilon_{\max{}}$ as the largest possible input noise into the system, we obtain the bound
\begin{gather}
	\alpha\left\|\sum_{k = 0}^{\memlen + \memlenf} \conmat^{k}\ffvec\widetilde{\epsilon}[k]\right\|_2^2 < \alpha \epsilon_{\max{}} \norm{\conmatev}_2 \left| \frac{\decvar}{1 - \decvar} \right| \nonumber
\end{gather}

\bibliographystyle{neuron}           
\bibliography{NNmem,RWL1,MPbib,CSbib,BPDNBIB,Dynamicbib,Kalmanbib,AudioBib,Adam_Bibliography2,refs,SIref}




\end{document}